\PassOptionsToPackage{unicode}{hyperref}
\PassOptionsToPackage{hyphens}{url}
\PassOptionsToPackage{dvipsnames,svgnames,x11names}{xcolor}
\documentclass[final,3p,10pt,authoryear]{elsarticle}

\AtBeginDocument{%
  \setlength{\abovedisplayskip}{6pt plus 2pt minus 2pt}%
  \setlength{\belowdisplayskip}{6pt plus 2pt minus 2pt}%
  \setlength{\abovedisplayshortskip}{3pt plus 2pt}%
  \setlength{\belowdisplayshortskip}{3pt plus 2pt minus 1pt}%
  \setlength{\jot}{2pt}%
}

\usepackage{amsmath,amsfonts,amssymb,mathtools,mathrsfs,bm}
\usepackage{iftex}
\ifPDFTeX
  \usepackage[T1]{fontenc}
  \usepackage[utf8]{inputenc}
  \usepackage{textcomp}
\else
  \usepackage{unicode-math}
  \defaultfontfeatures{Scale=MatchLowercase}
  \defaultfontfeatures[\rmfamily]{Ligatures=TeX,Scale=1}
\fi
\usepackage{lmodern}
\IfFileExists{upquote.sty}{\usepackage{upquote}}{}
\IfFileExists{microtype.sty}{%
  \usepackage[]{microtype}
  \UseMicrotypeSet[protrusion]{basicmath}
}{}

\usepackage{xcolor}
\makeatletter
\ifx\paragraph\undefined\else
  \let\oldparagraph\paragraph
  \renewcommand{\paragraph}{
    \@ifstar
      \xxxParagraphStar
      \xxxParagraphNoStar
  }
  \newcommand{\xxxParagraphStar}[1]{\oldparagraph*{#1}\mbox{}}
  \newcommand{\xxxParagraphNoStar}[1]{\oldparagraph{#1}\mbox{}}
\fi
\ifx\subparagraph\undefined\else
  \let\oldsubparagraph\subparagraph
  \renewcommand{\subparagraph}{
    \@ifstar
      \xxxSubParagraphStar
      \xxxSubParagraphNoStar
  }
  \newcommand{\xxxSubParagraphStar}[1]{\oldsubparagraph*{#1}\mbox{}}
  \newcommand{\xxxSubParagraphNoStar}[1]{\oldsubparagraph{#1}\mbox{}}
\fi
\makeatother

\usepackage{longtable,booktabs,array}
\usepackage{multirow}
\usepackage{calc}
\usepackage{etoolbox}
\makeatletter
\patchcmd\longtable{\par}{\if@noskipsec\mbox{}\fi\par}{}{}
\makeatother
\IfFileExists{footnotehyper.sty}{\usepackage{footnotehyper}}{\usepackage{footnote}}
\makesavenoteenv{longtable}

\usepackage{graphicx}
\makeatletter
\def\maxwidth{\ifdim\Gin@nat@width>\linewidth\linewidth\else\Gin@nat@width\fi}
\def\maxheight{\ifdim\Gin@nat@height>\textheight\textheight\else\Gin@nat@height\fi}
\makeatother
\setkeys{Gin}{width=\maxwidth,height=\maxheight,keepaspectratio}
\makeatletter
\def\fps@figure{htbp}
\makeatother

\usepackage{float}
\usepackage{caption,subcaption}
\usepackage{algorithm}
\usepackage{algpseudocode}
\usepackage{placeins}
\usepackage[hang,flushmargin]{footmisc}
\floatstyle{ruled}
\restylefloat{algorithm}
\algrenewcommand\alglinenumber[1]{\scriptsize #1:}

\ifLuaTeX
  \usepackage{selnolig}
\fi
\biboptions{round,authoryear}
\journal{Journal of Multivariate Analysis}

\IfFileExists{xurl.sty}{\usepackage{xurl}}{}
\usepackage[colorlinks=true,linkcolor=blue,citecolor=blue,urlcolor=blue,filecolor=Maroon]{hyperref}
\usepackage{bookmark}
\hypersetup{hypertexnames=false}

\newtheorem{prop}{Proposition}[section]

\newtheorem{theorem}{Theorem}[section]
\newtheorem{lemma}{Lemma}[section]

\newtheorem{corollary}{Corollary}[section]

\newcommand{\suppappendix}[1]{%
  \refstepcounter{section}%
  \section*{Appendix \thesection: #1}%
}

\hypersetup{
  pdftitle={A New Test for High-Dimensional Elliptical Distributions},
  pdfauthor={Wenrui Wu, Bingye Yang, Minghua Deng, Xu Guo},
  pdfkeywords={Goodness-of-fit testing; Elliptical distributions; Gaussian approximation; Multiplier bootstrap}
}

\begin{document}
\def\spacingset#1{\renewcommand{\baselinestretch}{#1}\small\normalsize}
\begin{frontmatter}
\title{A New Test for High-Dimensional Elliptical Distributions}
\author[pku]{Wenrui Wu}
\author[pku]{Bingye Yang}
\author[pku]{Minghua Deng}
\author[bnu]{Xu Guo\corref{cor1}}
\cortext[cor1]{Corresponding author.}
\ead{xustat12@bnu.edu.cn}
\affiliation[pku]{organization={School of Mathematical Sciences, Peking University},
                  city={Beijing}, country={China}}
\affiliation[bnu]{organization={School of Statistics, Beijing Normal University},
                  city={Beijing}, country={China}}

\begin{abstract}
\spacingset{1}
Elliptical distributions provide a flexible and widely used extension of multivariate normal distribution. They play a critical role in many statistical procedures when dealing with high-dimensional data. However, goodness-of-fit testing for elliptical distributions remains challenging when the dimension $p$ is comparable to or larger than the sample size $n$. In this work, we propose a new test for high-dimensional elliptical distributions. We establish high-dimensional Gaussian approximation for the test statistic under general correlation structures, allowing the dimension to grow as \(\log p=o(n^{1/14})\) under finite moment conditions. We further develop Gaussian multiplier bootstrap test procedure and prove its theoretical validity. Numerical studies demonstrate stable finite-sample behavior and favorable power against a range of alternatives. Applications to real datasets illustrate practical utility of the proposed test.


\end{abstract}
\end{frontmatter}

\section{Introduction}

Elliptical distributions constitute one of the most important extensions of the multivariate normal distribution. A centered random vector \( \textbf{x}\in \mathbb{R}^p\) with covariance matrix \(\Sigma\) is said to be elliptically distributed if it admits the stochastic representation
\[
        \textbf{x} \stackrel{d}{=} \xi \Sigma^{1/2} \textbf{u} ,
\]
where \(\textbf{u}\) is uniformly distributed on the unit sphere in \(\mathbb{R}^p\), \(\xi\) is a nonnegative radial random variable independent of \(\textbf{u}\), and \(\mathbb{E} \xi^2=p\). 
A systematic treatment of elliptical distributions is provided by \citet{FangKotzNg2018}. Elliptical distributions have become a standard distribution assumption in multivariate analysis and high-dimensional statistics. 
Representative applications include sufficient dimension reduction under elliptical predictor distributions \citep{Li2018,ChenZhangZhou2022}, 
variable selection and feature screening \citep{LiLiuLou2017, zhou2020model}, and high-dimensional testing \citep{wang2015high, cai2023tests, xu2025adjusted, zhao2026hypothesis}.
These examples illustrate that ellipticity often serves as a structural assumption underlying downstream statistical procedures, making it important to assess its plausibility before such methods are applied. Let \(\mathcal{P}\) denote the class of elliptical distributions. Given a sample \(X = (X_1, \ldots, X_n) \), where \((X_i)_{i=1}^n\) are i.i.d. \(p\)-dimensional random vectors drawn from a population distribution \(P\), we are interested in testing
 \begin{equation}
 \label{eppli_test}
     H_0: P \in \mathcal{P} \quad \text{versus} \quad H_1: P \notin \mathcal{P} .
 \end{equation} 

There is a large number of classical literature on testing spherical and elliptical symmetry in low dimensions. Classical approaches include tests based on projections, empirical processes, characteristic functions, moment restrictions, ranks, graphical diagnostics, and bootstrap approximations. For spherical symmetry, representative examples include
\citet{Baringhaus1991,KoltchinskiiLi1998,LiangFangHickernell2008, HenzeHlavkaMeintanis2014}. For elliptical symmetry, 
\citet{LiFangZhu1997} proposed diagnostic Q--Q probability plots based on statistics invariant under orthogonal transformations, 
\citet{ManzottiPerezQuiroz2002} considered tests based on the uniformity of standardized directions, 
\citet{ZhuNeuhaus2003} developed conditional tests that are exactly valid when the center and shape matrix are known, and \citet{HufferPark2007} proposed a chi-square-type statistic based on slicing. More recent developments include pseudo-Gaussian and optimal tests against skew-elliptical alternatives
\citep{CassartHallinPaindaveine2008,BabicGelbgrasHallinLey2021},
Kolmogorov--Smirnov-type tests based on characterizations of spherical symmetry
\citep{AlbisettiBalabdaouiHolzmann2020}, and diagnostic smooth tests
\citep{DucharmeDeMicheaux2020}. These methods have substantially advanced the classical theory of testing elliptical distributions. However, these methods were developed under fixed-dimensional asymptotics, and a number of them require inversion of the sample covariance matrix, which becomes unstable or infeasible when \(p\) is comparable to or exceeds \(n\), thereby limiting their direct applicability in high dimensions.

Modern datasets in genomics, finance, imaging, and signal processing often have a dimension comparable to or larger than the sample size. Even for testing multivariate normality, asymptotically valid high-dimensional procedures have only recently become available. Representative developments include the nearest-neighbor approach of \citet{ChenXia2023} and the aggregation test of moments of \citet{CuiZhang2026}. 
While for testing high-dimensional elliptical distributions, \citet{TangLi2024} proposed a kernel-embedding test based on the radial--directional characterization of ellipticity and established its consistency when the dimension $p$ is in the order of $n^{1/4}$. More recently, \citet{WangLopes2026} developed a novel high-dimensional goodness-of-fit test procedure for elliptical distributions based on two estimators of a common kurtosis parameter. 
Their procedure remains valid when \(p\) and \(n\) diverge proportionally.
The test, however, uses sample splitting to estimate an average marginal kurtosis and a common structural quantity. Consequently, it may have limited power against alternatives for which departures from ellipticity are not manifested through a discrepancy in average marginal kurtosis.

In this paper, we develop a kurtosis-based goodness-of-fit test for elliptical distributions in high dimensions. Motivated by \citet{WangLopes2026}, our starting point is the simple but informative implication that, under an elliptical distribution with nondegenerate marginal variances and finite fourth moments, all coordinate-wise marginal kurtoses
\[
        \kappa_j
        =
        \frac{\mathbb{E}X_{1j}^4}{(\mathbb{E}X_{1j}^2)^2},
        \qquad j=1,\ldots,p ,
\]
are identical. It is designed to detect sparse and localized departures from the moment structure implied by ellipticity. Methodically, our proposed procedure does not require estimation of the inverse population covariance matrix. Theoretically, we establish high-dimensional Gaussian approximation for the test statistic under general correlation structures. 

The main contribution of this paper is to establish theoretical validity of the proposed test procedure in genuinely high-dimensional regimes. To the best of our knowledge, this is the first high-dimensional goodness-of-fit test for elliptical distributions that allows \( \log p=o(n^{\alpha}) \) for some \(\alpha>0\). The analysis is challenging for two reasons. First, the test statistic involves the maximum of nonlinear and dependent kurtosis estimators, whose covariance structure depends on unknown correlation matrix of the data. Second, the fourth-powers of coordinates generally fail to satisfy sub-exponential tail under the null, even in the Gaussian case. Thus a key difficulty is that classical Gaussian approximation results do not apply directly to the proposed test statistics. To overcome these challenges, we first decompose the test statistic into a leading linear component and a nonlinear remainder. The nonlinear remainder is controlled using Rosenthal-type inequalities, whereas a Gaussian approximation is established for the leading linear component. We then construct Gaussian multiplier bootstrap based on this leading component to determine the critical value, thereby accommodating the unknown covariance structure. To dispense with sub-Gaussian tail assumptions, we establish a maximal moment inequality for elliptical distributions as a by-product, which may be of independent interest. The resulting test controls the asymptotic size under high-dimensional regimes that allow \(\log p=o(n^{1/14})\) under finite moment conditions. We establish the validity of the Gaussian multiplier bootstrap procedure and also conduct power analysis. 

The remainder of the paper is organized as follows. Section~\ref{sec_2} introduces the proposed test statistic and its high-dimensional Gaussian approximation. Section~\ref{sec_3} develops the multiplier-bootstrap test procedure and establishes its theoretical validity. Section~\ref{sec_4} reports simulation studies and real-data applications. Technical proofs, and other theoretical results are provided in Appendices~A--F.

\textit{Notations}. Let \((X_i)_{i=1}^n\) be i.i.d. \(p\)-dimensional observations drawn from a population distribution \(P\). For two positive arrays \(\{a_n\},\{b_n\}\), we write \(a_n = o(b_n)\) to mean \(a_n/b_n \to 0\) as \(n \to \infty\). We write \(a_n\lesssim b_n\) if \(a_n \leq C\cdot b_n\) for sufficiently large \(n\), where \(C\) is a positive constant. If \(a_n \lesssim b_n\) and \(b_n \lesssim a_n\), we write \(a_n \asymp b_n\). For a matrix $A=(a_{ij})$, we denote by $A^{\odot m}=(a_{ij}^{m})$ its \(m\)-th Hadamard power. The $L^p$ norm of a random variable $\xi$ is defined as $\|\xi\|_{L_p} = [\mathbb{E}(|\xi|^p)]^{1/p}$, and the entrywise $l_p$ norm of a real matrix $A=(a_{ij})$ is defined as $\|A\|_{p} = (\sum_{i,j}|a_{ij}|^p )^{1/p}$. We set $\| A \|_{\text{max}} = \max_{i,j}|a_{ij}|$ . Let \(\Sigma_{jj}\) denote the \(j\)th diagonal entry of the matrix \(\Sigma\), and let \(R=(r_{jk})_{p\times p}\) denote the correlation matrix corresponding to \(\Sigma\).
We use \(Z^{\Sigma}\) to denote the multivariate normal random vector \(Z^{\Sigma} \sim N_p(0,\Sigma)\), where \(\Sigma\) is its covariance matrix. Let \(\hat{\Sigma} = \frac{1}{n}\sum_{i=1}^nX_iX_i^{\top}\). We set $\text{med}_{1 \leq j \leq p}q_j$ to represent the median of $\{q_1, \cdots, q_p \}$. Let \(C,C_1,C_2,\ldots\) denote positive constants independent of \(n\) and \(p\), which may differ from line to line.

\section{Construction of test statistic and its high-dimensional Gaussian approximation}
\label{sec_2}

As discussed above, the goodness-of-fit problem in~\eqref{eppli_test} can be approached through the null implication that all coordinate-wise marginal kurtoses share a common value. A natural estimator of \(\kappa_j\) is
\[
\hat\kappa_j =\frac{\sum_{i=1}^n{X_{ij}^4}/n}{\big(\sum_{i=1}^n{X_{ij}^2/n}\big)^2}.
\]
Previous studies, such as \citet{WangLopes2026}, have considered aggregating the component-wise estimators of \(\kappa_j\) and comparing their average with an estimator of
\[
\kappa_c := \frac{3\bigl(\operatorname{var}(\lVert X_1\rVert_2^2)+\operatorname{tr}(\Sigma)^2\bigr)}
{2\operatorname{tr}(\Sigma^2)+\operatorname{tr}(\Sigma)^2}
,\]
based on the identity \( p^{-1}\sum_{j=1}^p\kappa_j = \kappa_c\) under \(H_0\). This average-based identity, however, fails to distinguish the null from alternatives under which the marginal kurtosis parameters are heterogeneous but their average remains equal to \(\kappa_c\). For instance, when \(p=2\), \(\kappa_1 = \kappa_c + 1\), and \(\kappa_2 = \kappa_c - 1\), we still have \(p^{-1}\sum_{j=1}^p \kappa_j = \kappa_c \), whereas \(\kappa_j = \kappa_c \) does not hold for all \(j\). This observation suggests that averaging-based procedures may lose power against certain non-elliptical alternatives with heterogeneous marginal kurtoses. Motivated by this consideration, we propose the following test statistic:
\begin{equation}
    \label{Test_statistic}
    T_n = \sqrt{n} \max_{1\leq j<l\leq p} \big|\hat\kappa_j - \hat\kappa_l \big|.
\end{equation}
Under the null hypothesis, \(T_n\) should be small. Under alternatives for which not all the marginal kurtoses are identical, at least one $\kappa_j-\kappa_l$ is nonzero, and $T_n$ would be large. A formal power analysis is provided in Section~\ref{sec_3}.

The asymptotic theory for~\eqref{Test_statistic} is based on a
Berry--Esseen bound for maxima of sums of high-dimensional random vectors.
To begin with, we impose the following conditions:
\begin{flalign*}
\textnormal{(A.1)}
&&
\mathbb E\xi_i^2=p,
\qquad
\operatorname{var}\bigg(
    \frac{\xi_i^2-p}{\sqrt p}
\bigg)
=
\tau+o(1),
&&
\phantom{\textnormal{(A.1)}}
\end{flalign*}
for some constant \(\tau\geq0\), and
\begin{flalign*}
\textnormal{(A.2)}
&&
\begin{aligned}
    b
    \leq
    \Sigma_{jj}
    \leq
    B
    <
    \infty,
    \qquad
    1\leq j\leq p,
    \\
    |r_{jl}|
    \leq
    r_0
    <
    1,
    \qquad
    1\leq j<l\leq p,
\end{aligned}
&&
\phantom{\textnormal{(A.2)}}
\end{flalign*}
for some constants \(B,b,r_0>0\). Condition~(A.1) imposes a mild
normalization on the radial variability, while Condition~(A.2) provides
uniform non-degenerate bounds on the marginal variances and all pairwise correlations. No assumption on the eigenvalues of \(\Sigma\) is required, so the covariance matrix may be singular. We further assume
\begin{flalign*}
\textnormal{(A.3)}
&&
\bigg\|
    \frac{\xi_i^2-p}{\sqrt p}
\bigg\|_{L_4}
=
o\bigl(p^{1/4}\bigr),
\qquad
\mathbb E\bigl(\xi_i^{2k}\bigr)
=
O\bigl(p^k\bigr),
\qquad
k=5,6,7,8.
&&
\phantom{\textnormal{(A.3)}}
\end{flalign*}
Similar moment conditions have also appeared in \citet{Jiang2019}. Condition (A.3) is weaker than the assumption imposed in \citet{WangLopes2026}:
\[
    \bigg\|
    \frac{\xi_i^2-p}{\sqrt p}
    \bigg\|_{L_8}
    =
    o\big(p^{1/4}\big).
\]
The radial moment conditions are satisfied by many commonly used radial distributions, including those induced by \(\chi_p^2\) and \(\text{Poisson}(p)\). We first analyze the asymptotic behavior of \(\hat\kappa_j - \kappa_j\). Let
\begin{equation}
r_k = \frac{\mathbb{E}({\xi}^{2k})}{\mathbb{E}\big(\| \textbf{z} \|_2^{2k}\big) } , \quad k\geq 1,
    \nonumber
\end{equation}
where $\textbf{z} \sim N_p(0, I_p)$, and let \[
    \varepsilon_{n,p}=\bigg\{
    \frac{\log^{11}(2pn)}{n}
    \bigg\}^{1/6} + \bigg\{
    \frac{\log^{7} (2pn)}{\sqrt n}
    \bigg\}^{1/3} .\] 
    Let \(\mathcal I_p=\{(j,l):1\leq j<l\leq p\}\) and \(q_p=|\mathcal I_p|=p(p-1)/2\). For matrices with \(q_p\) rows or columns, we use the pairs in \(\mathcal I_p\) as the corresponding indices. Thus, \(A_{(j,l),\cdot}\) and \(A_{\cdot,(j,l)}\) denote the row and column of \(A\) indexed by \((j,l)\) respectively. Define the \(q_p\times p\) matrix \(A_p\) by \((A_p)_{(j,l),\cdot}=e_j^\top-e_l^\top, (j,l)\in\mathcal I_p\), where \(e_j\) is the \(j\)th standard basis vector in \(\mathbb R^p\). The following theorem provides a nonasymptotic Gaussian approximation for \(T_n\). 

\begin{theorem}
\label{thm1}
    Under \(H_0\), assume that Conditions~(A.1)--(A.3) hold. Then, for all sufficiently large \(n,p\),
    \begin{equation}
        \label{thm1_cls1}
        \sup_{x\in\mathbb R}\bigg|
        \mathbb P\big(T_n>x\big)
        -\mathbb P\Big(
        \max_{(j,l)\in\mathcal I_p}|Z_{jl}^{\Gamma^*}|>x
        \Big)\bigg|
        \leq C\varepsilon_{n,p},
    \end{equation}
    where \(\Gamma^*=A_p\tilde\Sigma^*A_p^\top\), and
    \(\tilde{\Sigma}^{*}_{jk} = r_4 \big[ 24r_{jk}^4 + 72r_{jk}^2 + 9 \big] + r_2^3\big[ 72r_{jk}^2 + 36 \big] - 9r_2^2 - r_2r_3\big[ 144r_{jk}^2 + 36 \big]\)
    for \(j\neq k\), while
    \(\tilde{\Sigma}^{*}_{jj} = 105r_4-180r_2r_3 +108r_2^3 - 9r_2^2\).
    In particular, \(
        \Gamma^*_{(j,l),(k,m)}
        =\tilde\Sigma^*_{jk}-\tilde\Sigma^*_{jm}
        -\tilde\Sigma^*_{lk}+\tilde\Sigma^*_{lm}\) for any \((j,l),(k,m)\in\mathcal I_p\).
\end{theorem}
The radial moment ratios \(r_k\)'s converge to one under the stated conditions. This gives the following approximation with a covariance matrix determined by \(R\).
\begin{prop}
\label{prop_1}
    Under \(H_0\), assume the same conditions as in Theorem~\ref{thm1}. If \(\log p=o(n^{1/14})\), then, as \(n,p\to\infty\),
    \begin{equation}
        \label{cls2_prop1}
        \sup_{x\in\mathbb R}\bigg|
        \mathbb P\big(T_n>x\big)
        -\mathbb P\Big(
        \max_{(j,l)\in\mathcal I_p}|Z_{jl}^{\Gamma}|>x
        \Big)\bigg|
        \to0,
    \end{equation}
    where \(\Gamma=24A_pR^{\odot4}A_p^\top\), with \( \Gamma_{(j,l),(k,m)} =24(r_{jk}^4-r_{jm}^4-r_{lk}^4+r_{lm}^4) \) and \( \Gamma_{(j,l),(j,l)}=48(1-r_{jl}^4)\).
\end{prop}

\noindent\textbf{Remark.}
The proofs of Theorem~\ref{thm1} and Proposition~\ref{prop_1} proceed in three main steps. First, we decompose \(\hat{\kappa}_j -\kappa_j \) into a leading linear component and several nonlinear remainder terms. Second, we establish that these remainder terms are uniformly negligible at the \(n^{-1/2}\) scale under the stated moment conditions and \(\log p=o(n^{1/14})\). Third, we invoke high-dimensional Gaussian approximation results for maxima of sums of random vectors \citep{victor2017clt}, verify the requisite moment conditions, and characterize the corresponding covariance structure. A key technical difficulty is that fourth-power coordinate transformations generally do not satisfy dimension-free sub-exponential tail conditions. Indeed, even in the Gaussian case, $\mathbb{E}\{\exp(tX_{1j}^4)\} = \infty$ for all $t > 0$. Standard concentration inequalities for uniformly sub-exponential coordinates are therefore not directly applicable. To circumvent this difficulty, we establish a maximal moment inequality for elliptical distributions that controls \(\mathbb{E}(\max_{1\leq j\leq p} X_{1j}^{2k})\), which may be of independent interest.
\begin{lemma}
\label{lemma_elliptical_max_expect}
Assume \( \textbf{x}=\xi\,\Sigma^{1/2}\textbf{u}\), where \(\textbf{u}\sim \mathrm{Unif}(S^{p-1})\), \(\xi\perp \textbf{u}\),
\(\Sigma\succeq 0\), \(\max_{1\le j\le p}\Sigma_{jj}=O(1)\), and \(\mathbb E|\xi|^{2k}=O(p^k)\) for some \(k \geq 1\). Then,
\[
\mathbb E\Big[\max_{1\le j\le p} x_{j}^{2k}\Big]=O\big(\log^k p \big).
\]
\end{lemma}
This lemma provides an upper bound for the expected maximum of even coordinate-wise powers under an elliptical distribution. An intuitive example to check this property is the multivariate normal distribution. For $\textbf{x} \sim N_p(0, I_p)$, this rate is consistent with the familiar relation $\max_{1 \leq j \leq p}| \text{x}_{j}| = O_{\mathbb{P}}(\sqrt{\log p})$.


\section{Bootstrap test procedure and its validity}
\label{sec_3}

Although Proposition \ref{prop_1} establishes a Gaussian approximation for \(T_n\), its direct implementation is generally infeasible because the limiting covariance structure depends on the unknown correlation matrix \(R\). Following the multiplier-bootstrap framework of \citet{Peter1981,victor2013}, we employ a Gaussian multiplier bootstrap to approximate the distribution of \(\max_{(j,l)\in\mathcal I_p} |Z_{jl}^{\Gamma}|\).

In this section, we first introduce the bootstrap statistic and present the complete implementation of the proposed test procedure. We then establish the asymptotic validity of the test by showing that, under the null hypothesis, its rejection probability is controlled at the prespecified significance level \(\alpha\in(0,1)\). Finally, we investigate the power properties of the proposed test and demonstrate that it is consistent against a broad class of alternatives.

\subsection{Gaussian multiplier bootstrap}

To begin with, define the coordinate scores
\[
    h_{ij}
    =
    \Sigma_{jj}^{-2}
    \big(
    X_{ij}^{4}-6r_2\Sigma_{jj}X_{ij}^2
    \big),
    \qquad 1\leq j\leq p,
\]
and the random vector \(w_i=(w_{i,jl})_{(j,l)\in\mathcal I_p}\), where
\begin{equation}
    w_{i,jl}=h_{ij}-h_{il},
    \qquad (j,l)\in\mathcal I_p.
    \nonumber
\end{equation}
Its estimator is \(v_i=(v_{i,jl})_{(j,l)\in\mathcal I_p}\), with
\begin{equation}
    v_{i,jl}=\hat h_{ij}-\hat h_{il},
    \qquad
    \hat h_{ij}
    =
    \hat\Sigma_{jj}^{-2}
    \big(
    X_{ij}^{4}-2\hat\kappa_j\hat\Sigma_{jj}X_{ij}^2
    \big),
    \quad (j,l)\in\mathcal I_p.
    \nonumber
\end{equation}
Here $\hat\Sigma_{jj}$'s are sample estimator of $\Sigma_{jj}$'s. Let
\[
    \bar v=\frac1n\sum_{i=1}^n v_i,
    \qquad
    \hat\Omega
    =
    \frac1n\sum_{i=1}^n
    \big(v_i-\bar v\big)\big(v_i-\bar v\big)^\top.
\]
For \(1\leq j,k\leq p\), let
\[
    \hat r_{jk}
    =
    \frac{\hat\Sigma_{jk}}
    {\big(\hat\Sigma_{jj}\hat\Sigma_{kk}\big)^{1/2}},
    \qquad
    \hat\Sigma_{jk}=\frac1n\sum_{i=1}^n X_{ij}X_{ik}.
\]
Set
\[
    g_{ij}=h_{ij}-\frac1p\sum_{k=1}^p h_{ik},
    \qquad
    \hat g_{ij}=\hat h_{ij}-\frac1p\sum_{k=1}^p\hat h_{ik},
    \qquad
    \bar{\hat g}_j=\frac1n\sum_{i=1}^n\hat g_{ij},
\]
\[
    \hat\omega_j
    =24\bigg(
    1-\frac2p\sum_{k=1}^p\hat r_{jk}^4
    +\frac1{p^2}\sum_{k,l=1}^p\hat r_{kl}^4
    \bigg),
\]
\[
    \hat q_j
    =\begin{cases}
    \displaystyle
    \frac{n^{-1}\sum_{i=1}^n
    \big(\hat g_{ij}-\bar{\hat g}_j\big)^2}{\hat\omega_j},
    &\hat\omega_j>0,\\[2mm]
    1,&\hat\omega_j\leq0,
    \end{cases}
\]
and define
\[
    \mathcal G_{n,p}
    =
    \mathcal D_{n,p}
    \cap
    \bigg\{
    \min_{1\leq j\leq p}\hat\omega_j>0,
    \quad
    \operatorname{med}_{1\leq j\leq p}\hat q_j>0
    \bigg\},
\]
\[
    \mathcal D_{n,p} = \bigg\{
    \min_{1\leq j\leq p}\hat\Sigma_{jj}>0
    \bigg\}.
\]
We assume by default that the followings are evaluated on \(\mathcal D_{n,p}\). On \(\mathcal D_{n,p}^c\), we set \(\hat h_{ij}=\hat r_{jk}=0\), \(v_i=0\) and \(\hat a_{n,p}=1\), and do not reject, consistently with the convention \(T_n=0\).
Define the bootstrap variable
\begin{equation}
    S_n
    =
    \frac{1}{\sqrt n}
    \sum_{i=1}^n \hat a_{n,p}\big(w_i-\bar w\big)e_i,
    \qquad
    \bar w=\frac1n\sum_{i=1}^n w_i,
    \nonumber
\end{equation}
where \((e_i)_{i=1}^n\) are i.i.d. \(N(0,1)\) multipliers and
\begin{equation}
    \label{bootstrap_factor}
    \hat a_{n,p}
    =
    \begin{cases}
    \displaystyle
    \bigg\{
    \frac{\max_{1\leq j\leq p}\hat q_j}
    {\operatorname{med}_{1\leq j\leq p}\hat q_j}
    \bigg\}^{-\gamma_{n,p}},
    &
    \mathcal G_{n,p},
    \\[3mm]
    1,
    &
    \mathcal G_{n,p}^c.
    \end{cases}
\end{equation}
Here
\[
    \gamma_{n,p}
    =\max\bigg\{\frac1{12},
    \min\bigg\{\frac18,
    \frac1{10}+\frac{\log(p/n)}{40\log n}
    \bigg\}\bigg\}.
\]
Under \(H_0\) and \(\log p=o(n^{1/14})\), Conditions~(A.1)--(A.3) imply \(\mathbb P(\mathcal G_{n,p})\to1\). The second branch is just to make \(\hat a_{n,p}\) well defined. The variable \(S_n\) remains infeasible since \(w_i\) depends on unknown population quantities under $H_0$. In practice, we use the following multiplier bootstrap statistic
\begin{equation}
    \label{gaussian_boot}
    \hat S_n
    =
    \frac{1}{\sqrt n}
    \sum_{i=1}^n \hat a_{n,p}\big(v_i-\bar v\big)e_i.
\end{equation}
The factor $\hat a_{n,p}$ in \eqref{gaussian_boot} attenuates the influence of a small number of coordinates whose empirical eighth-moment variability is unusually large relative to the typical coordinate. It can be shown that $|\hat{a}_{n,p} - 1|$ is asymptotically negligible, but it turns out that retaining this term is beneficial at a finite-sample level. The form of \(v_{i,jl}\) arises from the leading linear term in the expansion of \(\hat\kappa_j-\hat\kappa_l\), with the marginal kurtosis \(\kappa_j\) replaced by \(\hat\kappa_j\) in each coordinate score. By repeatedly generating multiplier variables and recomputing \(\hat S_n\), the bootstrap distribution of \(\|\hat S_n\|_{\infty}\) can be used to approximate the distribution of \(\max_{(j,l)\in\mathcal I_p} \big|Z_{jl}^{\Gamma}\big|\).
Algorithm~\ref{algorithm:mytest} summarizes the resulting test procedure.

\begin{algorithm}[!htbp]
\spacingset{1}
\normalsize
\setlength{\abovedisplayskip}{3pt}
\setlength{\belowdisplayskip}{3pt}
\setlength{\abovedisplayshortskip}{2pt}
\setlength{\belowdisplayshortskip}{2pt}
\setlength{\jot}{2pt}
\caption{\small{Multiplier-bootstrap test procedure}}
\label{algorithm:mytest}
\begin{algorithmic}[1]
\State \textbf{Input:} observations \(\{X_i\}_{i=1}^n\), number of bootstrap iterations \(L\), significance level \(\alpha\).
\vspace{0.2mm}


\State Compute the marginal moments and construct the coordinate scores:
\[
    \hat\Sigma_{jj}
    =\frac1n\sum_{i=1}^n X_{ij}^2,
    \qquad
    \hat\kappa_j
    =\frac{n^{-1}\sum_{i=1}^n X_{ij}^4}{\hat\Sigma_{jj}^2},
    \qquad 1\leq j\leq p,
\]
\[
    \hat h_{ij}
    =
    \hat\Sigma_{jj}^{-2}
    \big(
    X_{ij}^{4}-2\hat\kappa_j\hat\Sigma_{jj}X_{ij}^2
    \big).
\]

\State Compute
\[
    \bar{\hat h}_j=\frac1n\sum_{i=1}^n\hat h_{ij},
    \qquad
    \hat g_{ij}=\hat h_{ij}-\frac1p\sum_{k=1}^p\hat h_{ik},
    \qquad
    \bar{\hat g}_j=\frac1n\sum_{i=1}^n\hat g_{ij},
\]
\[
    \hat r_{jk}
    =\frac{n^{-1}\sum_{i=1}^n X_{ij}X_{ik}}
    {\big(\hat\Sigma_{jj}\hat\Sigma_{kk}\big)^{1/2}},
    \qquad
    \hat\omega_j
    =24\bigg(
    1-\frac2p\sum_{k=1}^p\hat r_{jk}^4
    +\frac1{p^2}\sum_{k,l=1}^p\hat r_{kl}^4
    \bigg),
\]
\[
    \hat q_j
    =\begin{cases}
    \displaystyle
    \frac{n^{-1}\sum_{i=1}^n
    \big(\hat g_{ij}-\bar{\hat g}_j\big)^2}{\hat\omega_j},
    &\hat\omega_j>0,\\[2mm]
    1,&\hat\omega_j\leq0,
    \end{cases}
\]
\[
    \hat a_{n,p}
    =
    \begin{cases}
    \displaystyle
    \bigg\{
    \frac{\max_{1\leq j\leq p}\hat q_j}
    {\operatorname{med}_{1\leq j\leq p}\hat q_j}
    \bigg\}^{-\gamma_{n,p}},
    &
    \mathcal G_{n,p},
    \\[3mm]
    1,
    &
    \mathcal G_{n,p}^c.
    \end{cases}
\]

\[
    \gamma_{n,p}
    =\max\bigg\{\frac1{12},
    \min\bigg\{\frac18,
    \frac1{10}+\frac{\log(p/n)}{40\log n}
    \bigg\}\bigg\}.
\]

\State Compute the test statistic
\[
    T_n
    =
    \sqrt n\max_{1\leq j<l\leq p}
    |\hat\kappa_j-\hat\kappa_l|
    =\sqrt n\big(
    \max_{1\leq j\leq p}\hat\kappa_j
    -\min_{1\leq j\leq p}\hat\kappa_j
    \big),
\]

\For{\(b=1,\ldots,L\)}
    \State Generate bootstrap multipliers \((e_i^{(b)})_{i=1}^n\stackrel{\mathrm{i.i.d.}}{\sim}N(0,1)\).
    \State Calculate
    \[
        G_j^{(b)}
        =
        \frac{1}{\sqrt n}
        \sum_{i=1}^n
        \big(\hat h_{ij}-\bar{\hat h}_j\big)e_i^{(b)},
        \qquad 1\leq j\leq p.
    \]
    \State Set
    \[
        T_n^{(b)}=\hat a_{n,p}\big(
        \max_{1\leq j\leq p}G_j^{(b)}
        -\min_{1\leq j\leq p}G_j^{(b)}
        \big)=\big\|\hat S_n^{(b)}\big\|_{\infty}.
    \]
\EndFor

\State Take \(\hat c_{1-\alpha}\) as the \(100(1-\alpha)\)th sample percentile of \(\{T_n^{(1)},\ldots,T_n^{(L)}\}\).

\If{\(T_n>\hat c_{1-\alpha}\)}
    \State \Return \(1\) \Comment{Reject the null hypothesis.}
\Else
    \State \Return \(0\) \Comment{Do not reject the null hypothesis.}
\EndIf

\end{algorithmic}
\end{algorithm}

\subsection{The validity of Bootstrap procedure}

The following theorem shows that the conditional distribution of the bootstrap statistic consistently approximates the Gaussian distribution appearing in Proposition~\ref{prop_1}. 

\begin{theorem}
\label{thm2}
Under \(H_0\), assume the same conditions as in Theorem~\ref{thm1}. Let
\[
    \mathring S_n
    =
    \frac1{\sqrt n}
    \sum_{i=1}^n
    \big(w_i-\bar w\big)e_i.
\]
Then, for any \(\eta\in(0,e^{-1})\), with probability at least \(1-\eta\),
\begin{equation}
\label{cls1_thm2}
\begin{aligned}
&\sup_{x\in\mathbb R}
\bigg|
\mathbb P\Big(
\max_{(j,l)\in\mathcal I_p}|\mathring S_{n,jl}|>x
\,\big|\,
\{X_i\}_{i=1}^n
\Big)
-
\mathbb P\Big(
\max_{(j,l)\in\mathcal I_p}
\big|Z_{jl}^{\Gamma^*}\big|>x
\Big)
\bigg|                                                   \\
&\quad\leq
C\Bigg[
\bigg\{
\frac{\log^9(2pn)\log^2(1/\eta)}{n}
\bigg\}^{1/6}
+
\bigg\{
\frac{\log^7(2pn)}{\eta^{1/2}n^{1/2}}
\bigg\}^{1/3}
\Bigg]
.
\end{aligned}
\end{equation}
Moreover, under \(\log p=o(n^{1/14})\), the bootstrap statistic satisfies
\begin{equation}
\label{cls2_thm2}
\sup_{x\in\mathbb R}
\bigg|
\mathbb P\Big(
\max_{(j,l)\in\mathcal I_p}|\hat S_{n,jl}|>x
\,\big|\,
\{X_i\}_{i=1}^n
\Big)
-
\mathbb P\Big(
\max_{(j,l)\in\mathcal I_p}
\big|Z_{jl}^{\Gamma}\big|>x
\Big)
\bigg|
\xrightarrow{\mathbb P}0
\end{equation}
as \(n,p\to\infty\).
\end{theorem}


Combining this result with Proposition~\ref{prop_1} yields the validity of the bootstrap approximation for \(T_n\) itself and the asymptotic size control of the proposed bootstrap test. These conclusions are summarized in the following corollary.

\begin{corollary}
\label{cor_bootstrap_Tn}
Under \(H_0\), assume the conditions of Theorem~\ref{thm2} and \(\log p=o(n^{1/14})\). Let
\[
    \hat M_n
    =
    \max_{(j,l)\in\mathcal I_p}|\hat S_{n,jl}|.
\]
Then
\begin{equation}
\label{cor_bootstrap_Tn_cls1}
\sup_{x\in\mathbb R}
\Big|
\mathbb P\big(T_n>x\big)
-
\mathbb P\Big(
\hat M_n>x
\,\big|\,
\{X_i\}_{i=1}^n
\Big)
\Big|
\xrightarrow{\mathbb P}0.
\end{equation}
More precisely, for the auxiliary bootstrap variable \(\mathring S_n\) and
any \(\eta\in(0,e^{-1})\), with probability at least \(1-\eta\),
\begin{equation}
\label{cor_bootstrap_Tn_cls2}
\begin{aligned}
&\sup_{x\in\mathbb R}
\bigg|
\mathbb P\big(T_n>x\big)
-
\mathbb P\Big(
\max_{(j,l)\in\mathcal I_p}|\mathring S_{n,jl}|>x
\,\big|\,
\{X_i\}_{i=1}^n
\Big)
\bigg|                                                   \\
&\quad\leq
C\Bigg[
\varepsilon_{n,p}
+
\bigg\{
\frac{\log^9(2pn)\log^2(1/\eta)}{n}
\bigg\}^{1/6}
+
\bigg\{
\frac{\log^7(2pn)}{\eta^{1/2}\sqrt n}
\bigg\}^{1/3}
\Bigg]
.
\end{aligned}
\end{equation}
Consequently, for any fixed \(\alpha\in(0,1)\), if \(\hat c_{1-\alpha}\) denotes the conditional \((1-\alpha)\)-quantile of \(\hat M_n\), then the bootstrap test
\begin{equation}
    \label{proposed_test}
    \phi_n=1\big\{T_n>\hat c_{1-\alpha}\big\}
\end{equation}
has asymptotic size \(\alpha\), namely
\[
    \lim_{n\to\infty}
    \mathbb P_{H_0}\big(T_n>\hat c_{1-\alpha}\big)
    =
    \alpha.
\]
\end{corollary}

\noindent\textbf{Remark.}
A notable feature of test~\eqref{proposed_test} is that it accommodates the regime \(\log p=o(n^{1/14})\) under finite-moment conditions, without estimating the inverse sample covariance matrix or imposing sparsity or spectral restrictions on the population covariance matrix, subject to the pairwise nondegeneracy condition in (A.2). To the best of our knowledge, this is the first theoretically justified high-dimensional test for ellipticity that allows \(p\) to grow at an exponential rate in \(n\).

At a high level, the proofs of Theorem~\ref{thm2} and Corollary~\ref{cor_bootstrap_Tn} proceed in three stages. First, we establish uniform consistency of the empirical coordinate-score variances entering \(\hat q_j\), which gives \(|\hat a_{n,p}-1|\log^2p=o_{\mathbb P}(1)\). Second, we combine the stated moment conditions with Lemma~\ref{lemma_elliptical_max_expect} and apply the high-dimensional multiplier-bootstrap results of \citet{victor2017clt}. Third, we use Gaussian comparison and anti-concentration inequalities for maxima in \citet{chernozhukov2015comparison} to translate this nonasymptotic bound
into the limiting conditional distribution, which in turn yields the asymptotic size guarantee in Corollary~\ref{cor_bootstrap_Tn}.

\subsection{Power analysis}

The preceding results establish the validity of the proposed bootstrap
procedure under the null hypothesis. We next investigate its ability to
distinguish elliptical distributions from alternatives characterized by
heterogeneous marginal kurtoses. By the construction of \(T_n\), a natural
measure of signal strength is
\[
\Delta
:=
\max_{1\leq j<l\leq p}
\big|\kappa_j-\kappa_l\big|.
\]
Whenever the marginal kurtosis parameters are not all identical, the largest pairwise population contrast must be nonzero. We now formalize the
class of alternatives detectable by the proposed test. For constants
\(b,B>0\), let \(\mathcal H_{n,p}(b,B)\) denote the collection of
distributions \(P\) under \(H_1\) satisfying
\begin{flalign*}
\textnormal{(C.1)}
&&
0<b\leq\Sigma_{jj},
\qquad
\kappa_j\leq B<\infty,
\qquad
1\leq j\leq p,
&&
\phantom{\textnormal{(C.1)}}
\end{flalign*}
\begin{flalign*}
\textnormal{(C.2)}
&&
M_{16,n,p}
:=
\mathbb E\Big(
    \max_{1\leq j\leq p}X_{1j}^{16}
\Big)
=
o(n),
&&
\phantom{\textnormal{(C.2)}}
\end{flalign*}
and
\begin{flalign*}
\textnormal{(C.3)}
&&
\Delta
\gg
M_{16,n,p}^{1/4}
\sqrt{\frac{\log p}{n}},
&&
\phantom{\textnormal{(C.3)}}
\end{flalign*}
The following proposition establishes consistency over any sequence of
alternatives belonging to this region.

\begin{prop}
\label{power_analysis}
For any sequence of alternatives
\(P=P_{n,p}\in\mathcal H_{n,p}(b,B)\), the bootstrap test defined in~\eqref{proposed_test} satisfies
\begin{equation}
\label{power_cls}
\lim_{n\to\infty}
\mathbb P_{H_1}
\big(
T_n>\hat c_{1-\alpha}
\big)
=
1.
\end{equation}
\end{prop}

\noindent\textbf{Remark.}
Condition~(C.1) rules out degeneracy of the marginal scales and
imposes uniform control on the coordinate-wise kurtoses, whereas
Condition~(C.2) controls the tails of the fourth-power terms entering
the empirical kurtoses and their bootstrap linearization. Under
Conditions~(C.1)--(C.3), one can show that
\[
\sqrt n
\max_{1\leq j<l\leq p}
\big|
\hat\kappa_j-\hat\kappa_l
-
\big(\kappa_j-\kappa_l\big)
\big|
=
O_{\mathbb P}\big(
M_{16,n,p}^{1/4}\sqrt{\log p}
\big),
\]
and \[
\hat c_{1-\alpha}
=
O_{\mathbb P}\big(
M_{16,n,p}^{1/4}\sqrt{\log p}
\big).
\]
Consequently,
\[
\begin{aligned}
T_n
&\geq
\sqrt n\,\Delta
-
\sqrt n
\max_{1\leq j<l\leq p}
\big|
\hat\kappa_j-\hat\kappa_l
-
\big(\kappa_j-\kappa_l\big)
\big|                                                  \\
&=
\sqrt n\,\Delta
-
O_{\mathbb P}\big(
M_{16,n,p}^{1/4}\sqrt{\log p}
\big).
\end{aligned}
\]
Condition~(C.3) ensures that the population signal dominates both
the stochastic estimation error and the bootstrap critical value,
which yields \(
\mathbb P_{H_1}
\big(
T_n>\hat c_{1-\alpha}
\big)
\to1.
\)
Proposition~\ref{power_analysis} therefore highlights the sensitivity
of the proposed test procedure to sparse or localized departures
in marginal fourth moments.

\section{Numerical studies}
\label{sec_4}

In this section, we investigate the finite-sample performance of the proposed bootstrap test. We first examine its empirical size under a range of elliptical null distributions with different radial distributions, covariance structures, and dimension-to-sample-size ratios. We then study power under sparse heterogeneous alternatives and several complementary alternatives that distinguish local from global departures. Finally, we apply the proposed test to two real datasets.

\subsection{Assessment of empirical size}
\label{subsec:size}
We first examine the empirical size of our proposed test under the
elliptical null hypothesis. Data are generated according to
\[
    X_i=\xi_i\Sigma^{1/2}U_i,
    \qquad
    i=1,\ldots,n,
\]
where \(U_i\) is uniformly distributed on the unit sphere in
\(\mathbb R^p\), \(\xi_i\) is independent of \(U_i\), and
\(\mathbb E(\xi_i^2)=p\). We consider the following five choices for the
distribution of \(\xi_i^2\):
\[
\begin{array}{ll}
\text{(i).}
&
\xi_i^2\sim\chi_p^2,
\\[1mm]
\text{(ii).}
&
\xi_i^2\sim\mathrm{Gamma}(p/4,4),
\\[1mm]
\text{(iii).}
&
\xi_i^2=2pB_i,
\qquad
B_i\sim\mathrm{Beta}(p/2,p/2),
\\[1mm]
\text{(iv).}
&
\xi_i^2=p \Gamma_{1i}/ \Gamma_{2i},
\qquad
\Gamma_{1i}\sim\mathrm{Gamma}(p,1),
\quad
\Gamma_{2i}\sim\mathrm{Gamma}(p+1,1),
\\[1mm]
\text{(v).}
&
\xi_i^2=p+\sqrt p\,(\varepsilon_i+S_{i,p}),
\end{array}
\]
where \(\Gamma_{1i}\) and \(\Gamma_{2i}\) are independent, \(\varepsilon_i\) is a
Rademacher random variable independent of \(S_{i,p}\), and
\[
    S_{i,p} =
    \begin{cases}
        p^{1/2},
        &
        \text{with probability }p^{-3/2},
        \\[1mm]
        \displaystyle
        -\frac{ 1 }{p - p^{-1/2} } ,
        &
        \text{with probability }1 - p^{-3/2}.
    \end{cases}
\]
All five constructions satisfy \(\mathbb E(\xi_i^2)=p\) and
Conditions~(A.1)--(A.3). They include the Gaussian radial distribution, a
gamma-type radial distribution with increased radial variability, a bounded
beta-type distribution, a ratio-type beta-prime distribution, and a
rare-spike triangular-array distribution. The last construction (v) satisfies
the \(L_4\)-based condition (A.3) imposed in this paper but does not satisfy the
stronger requirement in \citet{WangLopes2026}.

For the covariance matrix \(\Sigma\), we consider four designs:
\[
\begin{array}{ll}
\text{(1).}
&
\Sigma=I_p,
\\[1mm]
\text{(2).}
&
\Sigma_{jk}=0.7^{|j-k|},
\\[1mm]
\text{(3).}
&
\begin{aligned}[t]
&\Sigma\text{ is block equicorrelated with block size }20
\\[-1mm]
&\text{ and within-block correlation }0.8,
\end{aligned}
\\[1mm]
\text{(4).}
&
\Sigma=0.1I_p+0.9\mathbf 1\mathbf 1^\top.
\end{array}
\]
The last two settings represent strong local and global correlation structures, respectively. In particular, design~(4) has strong pairwise correlation \(0.9\) between every two distinct coordinates.

The sample size is fixed at \(n=200\), and the dimension is selected from
\(p\in\{100,200,400\}\), corresponding to \(p/n\in\{0.50,1.00,2.00\}\). For each setting, we generate 1000 Monte Carlo datasets and use \(B=200\) Gaussian multiplier bootstrap replications. All rejection decisions are based directly on the bootstrap critical value in Algorithm~\ref{algorithm:mytest} at the nominal significance level \(\alpha=0.05\).

Table~\ref{tab:size_assessment} reports the empirical rejection probabilities
under the null hypothesis, in percentage form. Across the 60 settings, the
average empirical size is \(5.18\%\), and the median is \(5.20\%\). Among
the 60 entries, 42 lie between \(4\%\) and \(6\%\), and 47 lie in the exact
95\% binomial interval \([3.7\%,6.4\%]\) for a nominal rejection
probability of \(5\%\). The average rejection probabilities for
\(p/n=0.50\), \(1.00\), and \(2.00\) are \(5.81\%\), \(5.04\%\), and
\(4.70\%\), respectively.

The largest empirical size is \(7.9\%\), which occurs for the beta-prime
radial distribution under the block-equicorrelation design when
\(p/n=0.50\). The smallest empirical size is \(2.8\%\), occurring for the
beta-prime radial distribution under the identity covariance design when
\(p/n=2.00\). Under the strongest equicorrelation design, the empirical
sizes range from \(4.5\%\) to \(7.6\%\), with an average of \(5.82\%\).
Overall, the proposed multiplier-bootstrap procedure provides reasonably
stable size control across substantially different radial distributions and
under both weak and strong correlation structures.

\begin{table}[!htbp]
\centering
\caption{Empirical size of the proposed multiplier-bootstrap test at the
5\% nominal level. Entries are rejection probabilities $\times 100$, based on 1000 Monte Carlo replications with \(B=200\) bootstrap
replications.}
\label{tab:size_assessment}
\begin{tabular*}{\textwidth}{
@{\extracolsep{\fill}}
l
*{12}{c}
@{}
}
\toprule
&
\multicolumn{4}{c}{\(p/n=0.50\)}
&
\multicolumn{4}{c}{\(p/n=1.00\)}
&
\multicolumn{4}{c}{\(p/n=2.00\)}
\\
\cmidrule(lr){2-5}
\cmidrule(lr){6-9}
\cmidrule(lr){10-13}
&
(1)&(2)&(3)&(4)
&
(1)&(2)&(3)&(4)
&
(1)&(2)&(3)&(4)
\\
\midrule
(i)
&
4.0&5.3&7.5&6.3
&
3.5&4.4&6.4&6.6
&
3.2&3.5&4.5&5.3
\\
(ii)
&
5.1&5.6&7.8&5.3
&
4.2&4.1&5.8&5.2
&
4.2&4.0&5.2&6.9
\\
(iii)
&
5.7&4.6&6.0&5.4
&
4.3&5.1&5.6&4.5
&
3.9&4.7&3.8&7.6
\\
(iv)
&
5.2&4.1&7.9&6.4
&
5.0&5.2&5.7&5.0
&
2.8&4.2&5.4&5.9
\\
(v)
&
4.9&5.6&7.6&5.8
&
4.7&4.4&5.9&5.1
&
3.5&3.6&5.7&6.0
\\
\bottomrule
\end{tabular*}
\end{table}

\FloatBarrier

\subsection{Assessment of empirical power}
\label{subsec:power}

We next investigate the finite-sample power of the proposed test. For
comparison, we implement the high-dimensional test of
\citet{WangLopes2026} according to the procedure proposed in their paper.

Throughout this subsection, the sample size is \(n=200\), the dimension is
\(p\in\{100,200,400\}\), and the nominal significance level is
\(\alpha=0.05\). Each rejection probability is estimated from 1000 Monte
Carlo replications. For each Monte Carlo sample, the proposed test uses
\(B=200\) Gaussian multiplier bootstrap replications. Our objective is not
to claim uniform dominance over the competing methods, but rather to
demonstrate that the proposed test provides complementary power against
alternatives under which a small subset of coordinates exhibits distinct
marginal kurtosis.

\textit{Sparse design of alternatives.}
We first consider sparse heterogeneous alternatives. For
\(i=1,\ldots,n\), let
\[
    X_i=\Sigma^{1/2}G_i,
    \qquad
    G_i=\bigl(G_{i1},\ldots,G_{ip}\bigr)^\top,
\]
where the coordinates of \(G_i\) are independent and have unit variance.
For an active set \(\mathcal A\subset\{1,\ldots,p\}\), whose elements are
distributed approximately evenly across the coordinate indices, we generate
\begin{equation}
    G_{ij}
    =
    \begin{cases}
        \sqrt{1-h}\,Z_{ij}+\sqrt{h}\,H_{ij},
        & j\in\mathcal A,\\[1mm]
        Z_{ij},
        & j\notin\mathcal A,
    \end{cases}
    \nonumber
\end{equation}
where \(Z_{ij}\sim N(0,1)\), and \(H_{ij}\) is independent of \(Z_{ij}\)
and follows the rare-spike distribution
\[
    H_{ij}
    =
    \begin{cases}
        0,
        & \text{with probability }1-q,\\[1mm]
        q^{-1/2},
        & \text{with probability }q/2,\\[1mm]
        -q^{-1/2},
        & \text{with probability }q/2.
    \end{cases}
\]
When \(h=0\), the observations are Gaussian and hence elliptical. As \(h\)
increases, a small fraction of the coordinates departs from the Gaussian
baseline through changes in their marginal fourth moments. This construction
is in the spirit of rare and heterogeneous alternatives commonly considered
in high-dimensional testing problems \citep{DonohoJin2004,DonohoJin2015}. We use the same four covariance structures as in the size assessment, with the rare-spike parameters chosen as follows:
\[
\begin{array}{lll}
\text{(1).}
 \quad q=0.10, \quad|\mathcal A|/p=0.02,
\\[1mm]
\text{(2).}
 \quad q=0.05,\quad |\mathcal A|/p=0.01,
\\[1mm]
\text{(3).}
 \quad q=0.05,\quad |\mathcal A|/p=0.01,
\\[1mm]
\text{(4).}
 \quad q=0.01,\quad |\mathcal A|/p=0.02.
\end{array}
\]

Figures~\ref{fig:power_main_g1}--\ref{fig:power_main_g4} report the
empirical rejection probabilities. Under the identity and Toeplitz
covariance designs in Figures~\ref{fig:power_main_g1}
and~\ref{fig:power_main_g2}, respectively, the power of the proposed test
increases rapidly and approaches one once the signal becomes moderately
strong. The Wang--Lopes test also attains substantial power in these two
settings, although its increase is generally more gradual.

\begin{figure}[H]
    \centering
    \includegraphics[width=\textwidth]{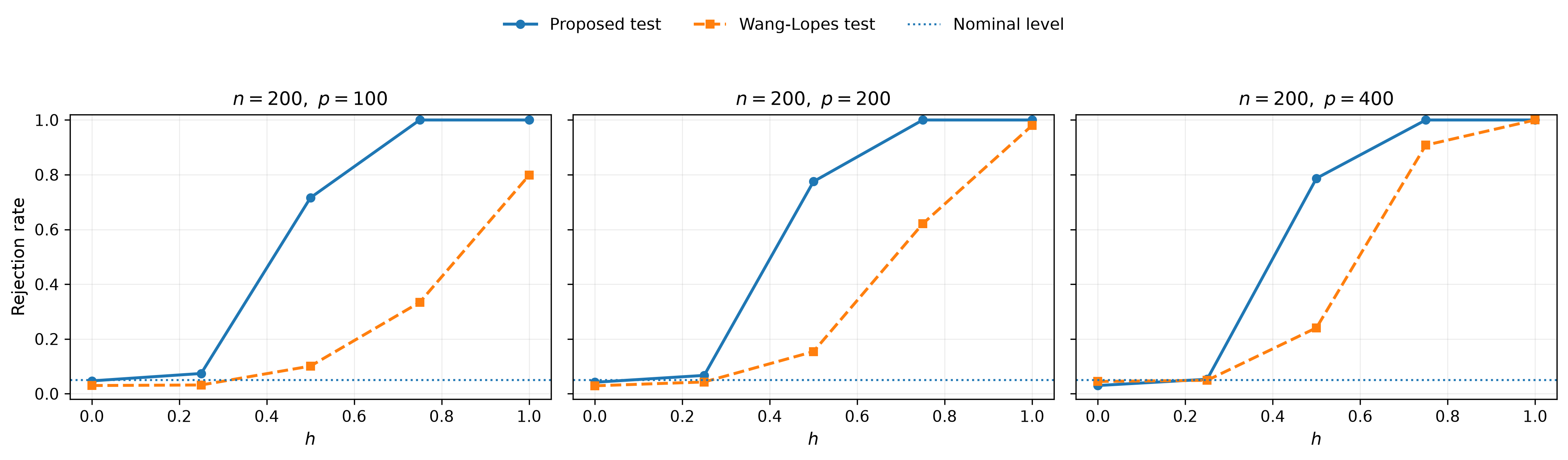}
    \caption{Empirical rejection probabilities under the sparse rare-spike
    alternative with identity covariance \(\Sigma=I_p\),
    \(q=0.10\), and \(|\mathcal A|/p=0.02\).}
    \label{fig:power_main_g1}
\end{figure}

\begin{figure}[H]
    \centering
    \includegraphics[width=\textwidth]{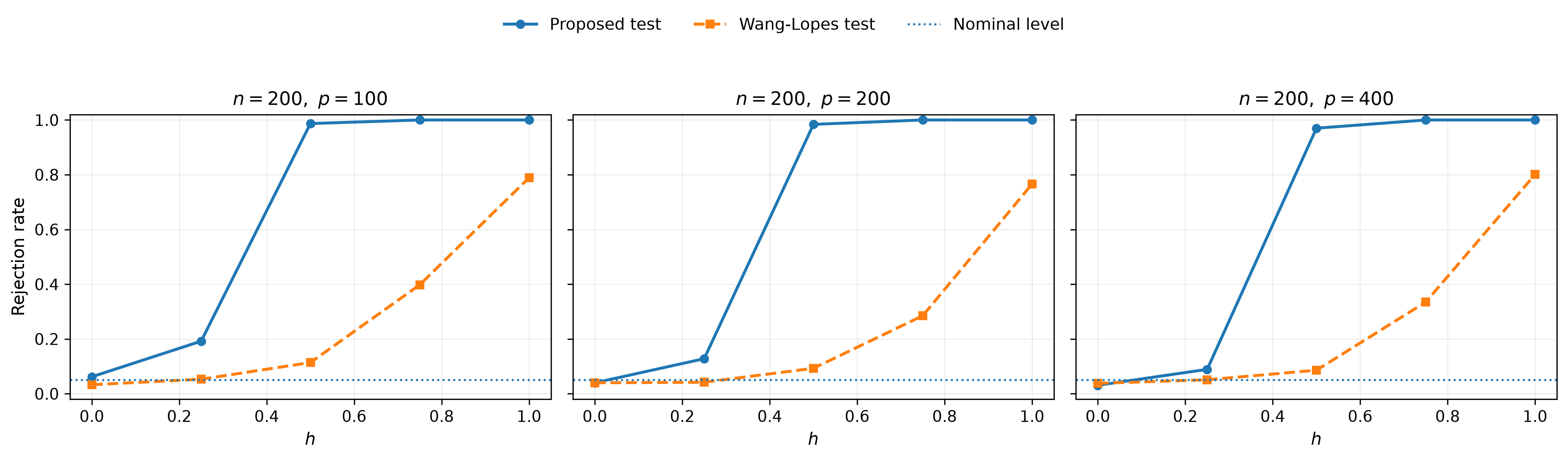}
    \caption{Empirical rejection probabilities under the sparse rare-spike
    alternative with Toeplitz covariance
    \(\Sigma_{jk}=0.7^{|j-k|}\), \(q=0.05\), and
    \(|\mathcal A|/p=0.01\).}
    \label{fig:power_main_g2}
\end{figure}

Under the block-equicorrelation and global-equicorrelation designs,
displayed in Figures~\ref{fig:power_main_g3} and~\ref{fig:power_main_g4}, respectively, the proposed procedure continues to exhibit steadily increasing power. Under block equicorrelation, its rejection probability at \(h=1\) ranges from approximately \(0.88\) to \(0.96\) as \(p\) increases from 100 to 400. Under the more challenging
global-equicorrelation design with pairwise correlation \(0.9\), the
corresponding rejection probabilities range from approximately \(0.85\) to
\(0.89\). In contrast, the rejection probability of the Wang--Lopes test
remains close to the nominal level in both settings. These results illustrate
the benefit of targeting the largest pairwise discrepancy in marginal kurtoses when the departure is sparse, including in the presence of strong cross-coordinate dependence.

\begin{figure}[H]
    \centering
    \includegraphics[width=\textwidth]{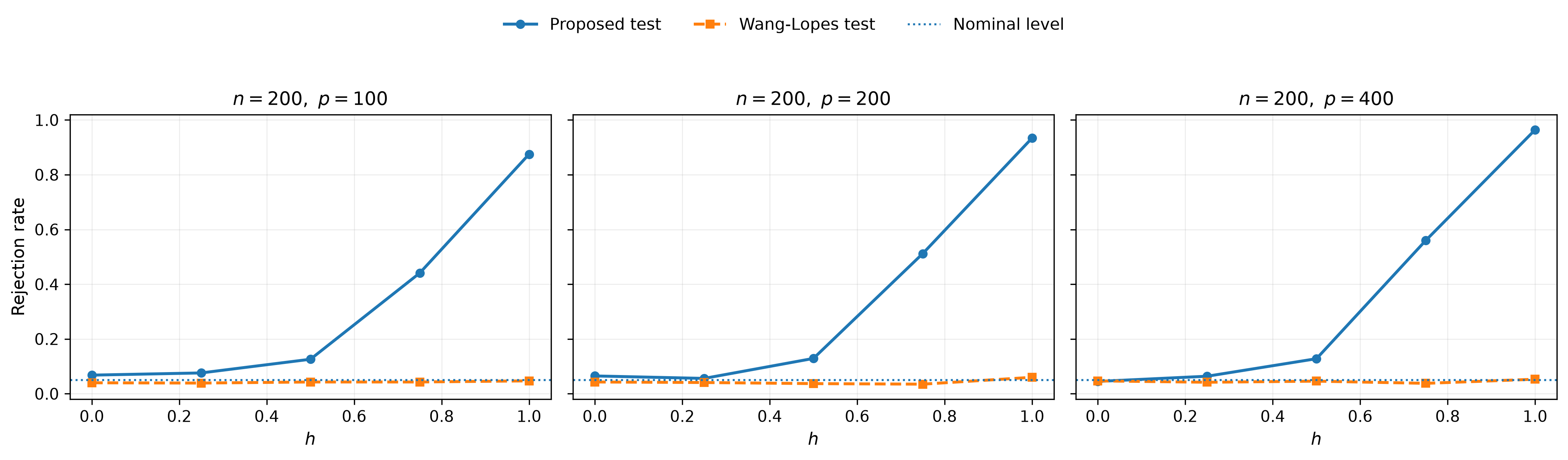}
    \caption{Empirical rejection probabilities under the sparse rare-spike
    alternative with block-equicorrelation covariance, block size \(20\),
    within-block correlation \(0.8\), \(q=0.05\), and
    \(|\mathcal A|/p=0.01\).}
    \label{fig:power_main_g3}
\end{figure}

\begin{figure}[H]
    \centering
    \includegraphics[width=\textwidth]{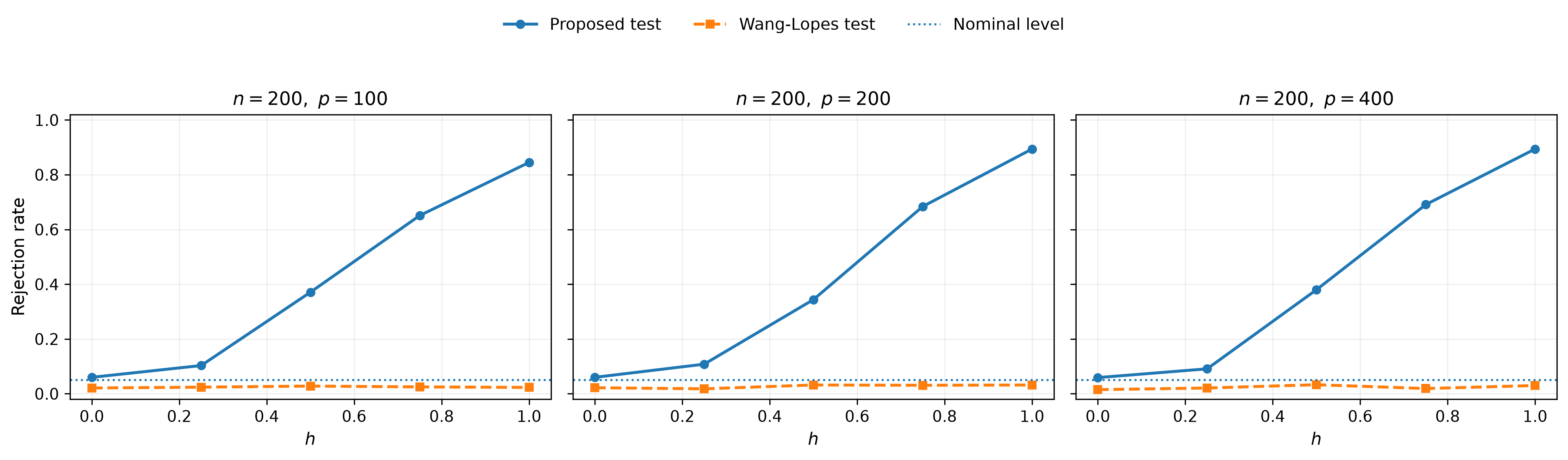}
    \caption{Empirical rejection probabilities under the sparse rare-spike
    alternative with equicorrelation covariance
    \(\Sigma=0.1I_p+0.9\mathbf 1\mathbf 1^\top\),
    \(q=0.01\), and \(|\mathcal A|/p=0.02\).}
    \label{fig:power_main_g4}
\end{figure}

\FloatBarrier

\textit{Several complementary cases.}
Section~S1 of the Supplementary Material reports additional power experiments under heterogeneous-kurtosis and standardized beta alternatives. The proposed test is effective against local kurtosis heterogeneity, whereas the Wang--Lopes test is more powerful in the standardized beta designs with generic rotations. These cases illustrate the complementary power of the two procedures.

\FloatBarrier

\subsection{Gene expression data}
\label{sec:gse22309_data}

Our first application revisits the gene-expression data studied by
\citet{wang2015high}. The experiment measured skeletal-muscle biopsies
from 15 patients with type 2 diabetes before and after insulin treatment
on the Affymetrix Human Genome U95A array \citep{Wu2007}. The public data
are available from \href{https://www.ncbi.nlm.nih.gov/geo/query/acc.cgi?acc=GSE22309}
{NCBI GEO} under accession GSE22309.
Following the earlier analysis, we average probes that map unambiguously
to the same gene and form the insulin-stimulated minus untreated
difference for each patient.

\citet{wang2015high} analyzed curated C2 gene sets but did not report the
MSigDB release used in their study. We use the archived C2 v6.0 definition
of \texttt{KIM\_HYPOXIA}, which contains genes upregulated under hypoxia
in the study of \citet{Kim2003}. Of its 25 genes, 23 are represented on
the U95A array
under the conservative probe annotation used here. They are retained in
their GMT order, and every gene is centered across the 15 paired
differences. The represented genes include several involved in glucose
metabolism, such as \textit{GYS1}, \textit{TPI1}, \textit{LDHA},
\textit{PFKP}, \textit{ALDOA}, \textit{ALDOC}, and \textit{SLC2A1}.
Thus the complete sample matrix has \(n=15\) and \(p=23\). 
We apply both tests to every dimension \(p\in\{17,18,\ldots,23\}\). The proposed test uses \(B=20{,}000\) Gaussian multiplier-bootstrap replications. 


Figure~\ref{fig:gse22309_pvalue_curves} shows that neither test rejects at
the \(5\%\) level for any displayed dimension, and the proposed \(p\)-value is uniformly smaller over this range. At \(p=23\), the \(p\)-values of the proposed test and
Wang--Lopes's test  are \(0.0847\) and \(0.1461\), respectively.
\citet{wang2015high} developed its main theory first under an elliptical
model and used these paired differences as a mean-vector application. 
Our conclusion is therefore compatible with that modeling framework.
\begin{figure}[H]
    \centering
    \includegraphics[width=0.70\textwidth]
        {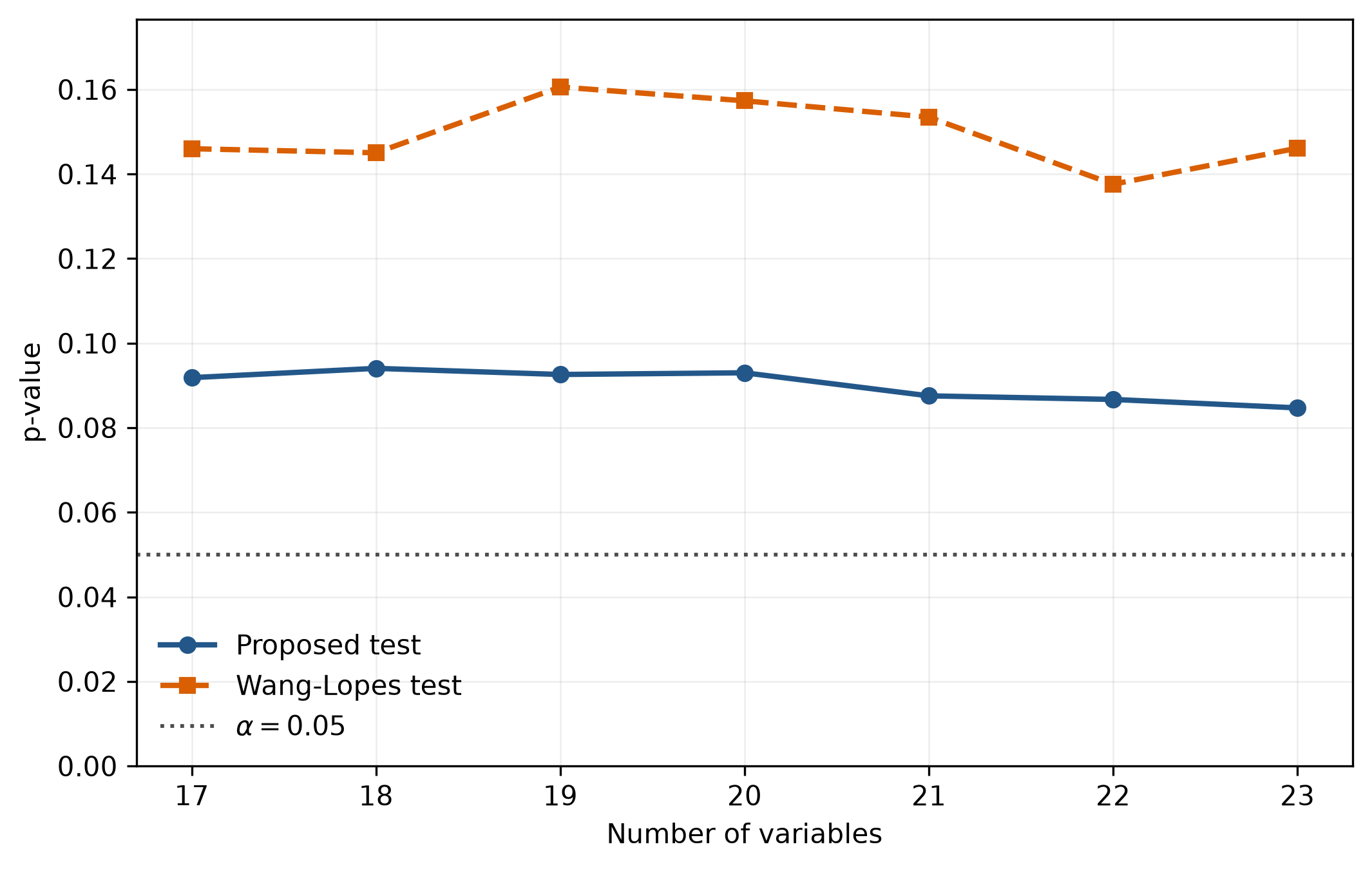}
    \caption{
        \(p\)-values of the proposed test and the Wang--Lopes test for the
        GSE22309 paired skeletal-muscle expression data as a function of the
        number of genes. The horizontal dotted line represents the nominal
        significance level \(0.05\).
    }
    \label{fig:gse22309_pvalue_curves}
\end{figure}

\subsection{Glass spectra data}
\label{sec:glass_data}

The second application uses the spectral data of \citet{Lemberge2000},
distributed as \texttt{data\_glass} in the
\href{https://search.r-project.org/CRAN/refmans/cellWise/html/data_glass.html}
{R package \texttt{cellWise}}.
The data contain 750 ordered spectral channels for $n=180$ archaeological-glass samples. Following the package vignette, we remove V1--V11, which are discrete, and V12--V13, which have zero or very small median absolute deviation. The remaining 737 channels are kept in their supplied order,
without further smoothing or scaling, and each channel is centered across samples. We evaluate both tests on the first \(p\) retained channels for
\(p\in\{200,250,300,\ldots,\allowbreak 650,700,737\}\). Here \(p>n\) at every point of the curve. The number of bootstrap replications is the same as in the first application. 

Figure~\ref{fig:glass_pvalue_curves} shows that both tests reject throughout this range.
The relative magnitude of the two \(p\)-values changes with the number of channels, but the rejection decision is stable over the displayed range. 


\begin{figure}[!htbp]
    \centering
    \includegraphics[width=0.70\textwidth]
        {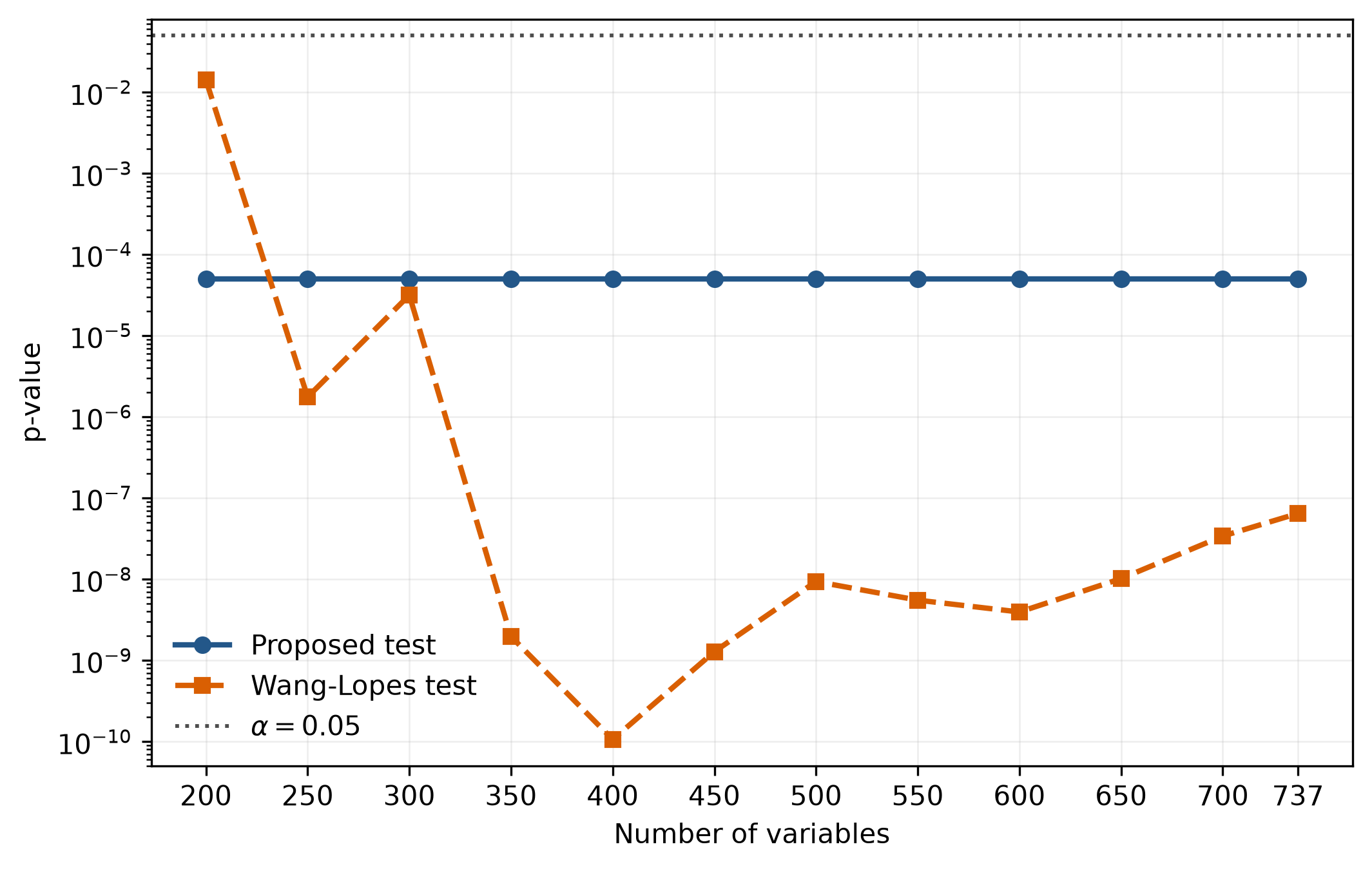}
    \caption{
        \(p\)-values of the proposed test and the Wang--Lopes test for the glass spectra data as a function of the number of spectral channels. The horizontal dotted line represents the nominal significance level \(0.05\).
    }
    \label{fig:glass_pvalue_curves}
\end{figure}


\FloatBarrier

\section{Discussion}

In this paper, we have developed a goodness-of-fit test for assessing ellipticity in high dimensions. The proposed procedure exploits the common kurtosis implication of elliptical distributions and does not require inversion of the sample covariance matrix or structural restrictions on the population correlation matrix. We establish a
high-dimensional Gaussian approximation for the test statistic under
finite-moment conditions, and construct a Gaussian multiplier bootstrap procedure for determining the critical value, allowing \(\log p=o(n^{1/14})\). A maximal moment inequality for elliptical distributions is obtained as a by-product of the
theoretical analysis. The numerical results indicate stable finite-sample
size control across a broad range of radial distributions and correlation
structures, together with favorable power against sparse and localized
departures in marginal kurtosis.


An extension concerns the quadratic statistic
\[
T_n^{*}
=
\frac{2n}{p(p-1)}
\sum_{1\leq j<l\leq p}
\big(
\widehat{\kappa}_j-\widehat{\kappa}_l
\big)^2.
\]
Unlike \(T_n\), which is constructed from a maximum of linear contrasts, \(T_n^{*}\) aggregates squared pairwise contrasts and is therefore a quadratic form in the marginal kurtosis estimators. Consequently, the Gaussian approximation for maxima of linear contrasts established in this paper does not directly yield the null distribution of \(T_n^{*}\). For fixed $p$, the arguments in this paper suggest a weighted chi-square limiting distribution. When $p$ also
diverges, an appropriately centered and scaled version of $\sqrt p T_n^{*}$ may admit a Gaussian limit
under additional conditions. Establishing such a result, together with a valid bootstrap
calibration, constitutes an interesting direction for future research.

Further as illustrated in our numerical studies, our proposed test and the Wang-Lopes's test are complementary to each other. It is then of interest to develop a procedure further combine their advantages. We will explore these possible topics in near future.

\section{Disclosure statement}

The authors report there are no competing interests to declare.

\section{Data availability statement}

The GSE22309 skeletal-muscle data and the archaeological-glass spectra data are publicly available from the sources cited in the Numerical studies section. 

\phantomsection\label{supplementary-material}
\bigskip

\begin{center}

{\large\bf SUPPLEMENTARY MATERIAL}

\end{center}

\begin{description}
\item[Supplementary material:]
Additional results of power assessment are provided in the Supplementary Material. 
\end{description}

\appendix
\renewcommand{\thesection}{\Alph{section}}
\suppappendix{Proof of Theorem~2.1}
\label{appendix_a}

All empirical ratios are evaluated on \(\mathcal D_{n,p}=\{\min_j\hat\Sigma_{jj}>0\}\), with \(\hat\kappa_j=0\) and the convention in Section~3 on \(\mathcal D_{n,p}^c\). The probability of this exceptional event tends to zero by \eqref{eq:sample_variance_rate}. Throughout the proof, Conditions~(A.1)--(A.3) are in force. The representation \(X_i=\xi_i\Sigma^{1/2}u_i\), with \(u_i=z_i/\|z_i\|_2\) and \(z_i\sim N_p(0,I_p)\) independent of \(\xi_i\), gives
\begin{equation}
\label{moment_xij}
    \mathbb E(X_{ij}^{2k})=(2k-1)!!r_k\Sigma_{jj}^{k},
    \qquad
    r_k=\frac{\mathbb E(\xi_i^{2k})}{\mathbb E\|z_i\|_2^{2k}}.
\end{equation}
We first record three auxiliary lemmas.

\begin{lemma}
\label{lemma_1}
Under Conditions~(A.1)--(A.3), \(r_k=1+o(p^{-3/4})\) for \(k=1,2,3,4\).
\end{lemma}
\noindent\textit{Proof.}
For \(D=\xi_i^2-p\), the assumptions give \(\mathbb ED=0\), \(\mathbb ED^2=O(p)\), and \(\mathbb E|D|^4=o(p^3)\). The Cauchy--Schwarz inequality yields \(\mathbb E|D|^3=o(p^2)\). Expanding \(\mathbb E(p+D)^k\) gives \(p^k+O(p^{k-1})+o(p^{k-3/4})\) for \(k=2,3,4\). Dividing by \(\mathbb E\|z_i\|_2^{2k}=p^k+O(p^{k-1})\) proves the assertion. \(\hfill\square\)

\begin{lemma}
\label{lemma_gaussian_comparison}
Let \(G^A,G^B\) be centered Gaussian vectors in \(\mathbb R^d\) with covariance matrices \(A=(a_{jk})\), \(B=(b_{jk})\), and put \(\Delta(A,B)=\max_{j,k}|a_{jk}-b_{jk}|\). If their coordinate variances lie in \([\underline\sigma^2,\overline\sigma^2]\), where \(0<\underline\sigma\leq\overline\sigma<\infty\), then
\[
\begin{aligned}
&\sup_{x\in\mathbb R}\bigg|
\mathbb P\Big(\max_{j\leq d}G^A_j\leq x\Big)
-\mathbb P\Big(\max_{j\leq d}G^B_j\leq x\Big)\bigg|\\
&\quad\leq C\Delta(A,B)^{1/3}
\bigg\{1\vee\log\bigg(\frac d{\Delta(A,B)}\bigg)\bigg\}^{2/3},
\end{aligned}
\]
where \(C\) depends only on \(\underline\sigma,\overline\sigma\), and the bound is zero when \(\Delta(A,B)=0\).
\end{lemma}
\noindent\textit{Proof.}
This is the Gaussian comparison inequality for maxima in \cite{chernozhukov2015comparison}. In particular, \(\Delta(A,B)\log^2d\to0\) implies convergence to zero in Kolmogorov distance. \(\hfill\square\)

\begin{lemma}
\label{lemma_2}
Under Conditions~(A.1) and~(A.3), for \(1\leq j,k\leq p\),
\begin{align*}
\operatorname{cov}\big(X_{1j}^4,X_{1k}^4\big)
&=r_4\big(72\Sigma_{jj}\Sigma_{kk}\Sigma_{jk}^2+24\Sigma_{jk}^4\big)
+9\big(r_4-r_2^2\big)\Sigma_{jj}^2\Sigma_{kk}^2,\\
\operatorname{cov}\big(X_{1j}^2,X_{1k}^2\big)
&=2r_2\Sigma_{jk}^2+(r_2-1)\Sigma_{jj}\Sigma_{kk},\\
\operatorname{cov}\big(X_{1k}^2,X_{1j}^4\big)
&=12r_3\Sigma_{jj}\Sigma_{jk}^2+3(r_3-r_2)\Sigma_{jj}^2\Sigma_{kk}.
\end{align*}
\end{lemma}
\noindent\textit{Proof.}
For a homogeneous polynomial of degree \(2s\), its expectation at \(X_1\) equals \(r_s\) times its expectation at \(\Sigma^{1/2}z_1\). Isserlis' theorem gives the Gaussian product moments
\[
\begin{aligned}
\mathbb E(Y_j^4Y_k^4)&=9\Sigma_{jj}^2\Sigma_{kk}^2
+72\Sigma_{jj}\Sigma_{kk}\Sigma_{jk}^2+24\Sigma_{jk}^4,\\
\mathbb E(Y_j^2Y_k^2)&=\Sigma_{jj}\Sigma_{kk}+2\Sigma_{jk}^2,
\qquad
\mathbb E(Y_j^4Y_k^2)=3\Sigma_{jj}^2\Sigma_{kk}+12\Sigma_{jj}\Sigma_{jk}^2,
\end{aligned}
\]
where \(Y=\Sigma^{1/2}z_1\). Multiplying by \(r_4,r_2,r_3\), respectively, and subtracting the products of marginal moments proves the formulas. See also Lemma~D.6 of \cite{WangLopes2026} and \cite{magnus1978moments}. \(\hfill\square\)

By \eqref{moment_xij}, \(\kappa_j=3r_2\). Put \(d_j=\hat\Sigma_{jj}-\Sigma_{jj}\), \(w_{ij}=X_{ij}^4-6r_2\Sigma_{jj}X_{ij}^2\), and
\begin{equation}
\label{Wj}
    W_j=\frac1n\sum_{i=1}^n(w_{ij}-\mathbb Ew_{ij}),
    \qquad
    V_j=-3r_2d_j^2,
    \qquad
    \delta_j=\frac1{\hat\Sigma_{jj}^2}-\frac1{\Sigma_{jj}^2}.
\end{equation}
Elementary algebra gives
\begin{equation}
\label{decomp_first}
    \hat\kappa_j-\kappa_j=\frac{W_j}{\Sigma_{jj}^2}+R_{n,j},
    \qquad
    R_{n,j}=\delta_j(W_j+V_j)+\frac{V_j}{\Sigma_{jj}^2}.
\end{equation}
On \(\mathcal A_n=\{\Delta_{\Sigma,n}:=\max_j|d_j|\leq b/2\}\), we have \(\max_j|\delta_j|\leq C\Delta_{\Sigma,n}\) and \(\max_j|V_j|\leq C\Delta_{\Sigma,n}^2\). Let \(L_p=\log(2p)\). The maximal sample-mean inequality \eqref{eq:maximal_sample_mean} and Lemma~\ref{lemma_elliptical_max_expect} imply
\begin{equation}
\label{eq:appendix_a_maximal_rates}
    \mathbb E\Delta_{\Sigma,n}\leq C\frac{L_p^{3/2}}{\sqrt n},
    \qquad
    \mathbb E\max_j|W_j|\leq C\frac{L_p^{5/2}}{\sqrt n}.
\end{equation}
Put \(\eta_{n,p}=(L_p^9/n)^{1/6}\). If \(\eta_{n,p}\) is bounded below by a fixed positive constant, the theorem follows by enlarging \(C\), since \(\varepsilon_{n,p}\geq\eta_{n,p}\). Otherwise, Markov's inequality gives, with probability at least \(1-C\eta_{n,p}\),
\[
    \Delta_{\Sigma,n}\leq\frac{CL_p^{3/2}}{\eta_{n,p}\sqrt n},
    \qquad
    \max_j|W_j|\leq\frac{CL_p^{5/2}}{\eta_{n,p}\sqrt n}.
\]
The first bound is \(O(n^{-1/3})\), so this event is contained in \(\mathcal A_n\) for sufficiently large \(n\). Consequently,
\begin{equation}
\label{eq:uniform_linearization_rate}
\begin{aligned}
    \mathcal E_{n,p}:=\sqrt n\max_j|R_{n,j}|
    &\leq C\sqrt n\big\{
    \Delta_{\Sigma,n}\max_j|W_j|+\Delta_{\Sigma,n}^2+\Delta_{\Sigma,n}^3\big\}
    \leq C\frac{L_p}{n^{1/6}},\\
    \mathbb P\big(\mathcal E_{n,p}>CL_p/n^{1/6}\big)
    &\leq C\eta_{n,p}.
\end{aligned}
\end{equation}

Define \(\psi_{ij}=(w_{ij}-\mathbb Ew_{ij})/\Sigma_{jj}^2\) and \(\psi_i=(\psi_{i1},\ldots,\psi_{ip})^\top\). Lemma~\ref{lemma_2} gives \(\operatorname{cov}(\psi_i)=\tilde\Sigma^*\), with
\begin{equation}
\label{eq:tilde_sigma_p_offdiag}
    \tilde\Sigma^*_{jk}
    =r_4\big(24r_{jk}^4+72r_{jk}^2+9\big)
    +r_2^3\big(72r_{jk}^2+36\big)-9r_2^2
    -r_2r_3\big(144r_{jk}^2+36\big),\quad j\ne k,
\end{equation}
\begin{equation}
\label{eq:tilde_sigma_p_diag}
    \tilde\Sigma^*_{jj}=105r_4-180r_2r_3+108r_2^3-9r_2^2.
\end{equation}
By Lemma~\ref{lemma_1} and \(|r_{jk}|\leq1\),
\[
    \Delta_p:=\max_{j,k}|\tilde\Sigma^*_{jk}-24r_{jk}^4|
    \leq C\big(|r_4-1|+|r_2r_3-1|+|r_2^3-1|+|r_2^2-1|\big)
    =o(p^{-3/4}).
\]
For \(\widehat K=(\hat\kappa_1,\ldots,\hat\kappa_p)^\top\), \eqref{decomp_first} yields
\begin{equation}
\label{eq:parallel_decomposition}
    A_p\sqrt n(\widehat K-3r_2\mathbf1_p)
    =\frac1{\sqrt n}\sum_{i=1}^n A_p\psi_i+A_p\mathcal R_{n,p},
    \qquad
    \mathcal R_{n,p}=\sqrt n(R_{n,1},\ldots,R_{n,p})^\top.
\end{equation}
Each row of \(A_p\) has \(\ell_1\)-norm two, so
\begin{equation}
\label{eq:parallel_remainder_rate}
    \|A_p\mathcal R_{n,p}\|_\infty\leq2\mathcal E_{n,p},
    \qquad
    \|A_p\psi_i\|_\infty\leq2\|\psi_i\|_\infty.
\end{equation}
The covariance of \(A_p\psi_i\) is \(\Gamma^*=A_p\tilde\Sigma^*A_p^\top\), whose entries are the four-term contrasts stated in Theorem~2.1. In particular,
\begin{equation}
\label{eq:breve_sigma_p_diag}
    \Gamma^*_{(j,l),(j,l)}
    =210r_4-360r_2r_3+216r_2^3-18r_2^2-2\tilde\Sigma^*_{jl}
    =48(1-r_{jl}^4)+o(p^{-3/4})
\end{equation}
uniformly in \((j,l)\). Condition~(A.2) therefore gives, for sufficiently large \(p\),
\begin{equation}
\label{nondegenerate_contrasts}
    0<c_0\leq\min_{(j,l)\in\mathcal I_p}\Gamma^*_{(j,l),(j,l)}
    \leq\max_{(j,l)\in\mathcal I_p}\Gamma^*_{(j,l),(j,l)}\leq C_0<\infty.
\end{equation}
Furthermore, Conditions~(A.2)--(A.3), \eqref{moment_xij}, and Lemma~\ref{lemma_elliptical_max_expect} imply
\[
    \max_{(j,l)\in\mathcal I_p}\mathbb E|(A_p\psi_i)_{jl}|^{2+k}=O(1),\quad k=1,2,
    \qquad
    \mathbb E\|A_p\psi_i\|_\infty^4\leq CL_p^8.
\]
Thus the variance, moment and polynomial-tail conditions of Proposition~2.1 in \cite{victor2017clt} hold with \(B_n=CL_p^2\) and tail exponent four. Applied to the positive and negative coordinates of \(A_p\psi_i\), in dimension \(2q_p\), that proposition gives
\begin{equation}
\label{eq:asymp_normal_parallel_W}
\begin{aligned}
&\sup_{x\in\mathbb R}\bigg|
\mathbb P\bigg(\bigg\|\frac1{\sqrt n}\sum_{i=1}^n A_p\psi_i\bigg\|_\infty>x\bigg)
-\mathbb P\Big(\max_{(j,l)\in\mathcal I_p}|Z_{jl}^{\Gamma^*}|>x\Big)\bigg|\\
&\quad\leq C\bigg[
\bigg\{\frac{B_n^2\log^7(2q_pn)}n\bigg\}^{1/6}
+\bigg\{\frac{B_n^2\log^3(2q_pn)}{\sqrt n}\bigg\}^{1/3}\bigg]
\leq C\varepsilon_{n,p},
\end{aligned}
\end{equation}
since \(q_p\leq p^2\). Put \(\delta_{n,p}=CL_p/n^{1/6}\). By \eqref{eq:uniform_linearization_rate} and \eqref{eq:parallel_remainder_rate}, the contrast remainder exceeds \(2\delta_{n,p}\) with probability at most \(C\eta_{n,p}\). Nazarov's inequality (Lemma~A.1 of \cite{victor2017clt}) and \eqref{nondegenerate_contrasts} give
\[
    \sup_{x\in\mathbb R}\mathbb P\bigg(
    \bigg|\max_{(j,l)\in\mathcal I_p}|Z_{jl}^{\Gamma^*}|-x\bigg|
    \leq2\delta_{n,p}\bigg)
    \leq C\delta_{n,p}\sqrt{\log(2q_p)}\leq C\eta_{n,p}.
\]
Since \(T_n=\|A_p\sqrt n(\widehat K-3r_2\mathbf1_p)\|_\infty\), the last two bounds and \eqref{eq:asymp_normal_parallel_W} imply
\begin{equation}
\label{eq:parallel_gaussian_p}
    \sup_{x\in\mathbb R}\bigg|
    \mathbb P(T_n>x)
    -\mathbb P\Big(\max_{(j,l)\in\mathcal I_p}|Z_{jl}^{\Gamma^*}|>x\Big)
    \bigg|\leq C\varepsilon_{n,p},
\end{equation}
which proves Theorem~2.1. \(\hfill\square\)

\suppappendix{Proof of Proposition~2.1}
\label{appendix_b}

Appendix A established
\[
    \Delta_p=\max_{1\leq j,k\leq p}
    |\tilde\Sigma^*_{jk}-24r_{jk}^4|=o(p^{-3/4}).
\]
Set \(\Gamma=24A_pR^{\odot4}A_p^\top\). By the row-norm bound for \(A_p\),
\[
    \Delta_p^A:=\max_{a,b\in\mathcal I_p}|\Gamma^*_{ab}-\Gamma_{ab}|
    \leq4\Delta_p=o(p^{-3/4}).
\]
The covariance entries are
\[
    \Gamma_{(j,l),(k,m)}
    =24\big(r_{jk}^4-r_{jm}^4-r_{lk}^4+r_{lm}^4\big),
    \qquad \Gamma_{(j,l),(j,l)}=48(1-r_{jl}^4).
\]
These variances are bounded above and away from zero by Condition~(A.2). The Schur product theorem implies that \(R^{\odot4}\), and hence \(\Gamma\), is positive semidefinite. The Gaussian comparison inequality in Lemma~\ref{lemma_gaussian_comparison}, applied to the augmented positive and negative coordinates, gives
\[
\begin{aligned}
&\sup_{x\in\mathbb R}\bigg|
\mathbb P\Big(\max_{(j,l)\in\mathcal I_p}|Z_{jl}^{\Gamma^*}|>x\Big)
-\mathbb P\Big(\max_{(j,l)\in\mathcal I_p}|Z_{jl}^{\Gamma}|>x\Big)
\bigg|\\
&\quad\leq C(\Delta_p^A)^{1/3}
\bigg\{1\vee\log\bigg(\frac{2q_p}{\Delta_p^A}\bigg)\bigg\}^{2/3}
=o(1).
\end{aligned}
\]
The right-hand side is understood as zero when \(\Delta_p^A=0\). Under \(\log p=o(n^{1/14})\), \(\varepsilon_{n,p}\to0\). Combining the last display with \eqref{eq:parallel_gaussian_p} proves the claimed Gaussian approximation.

Finally, if \(Z^{\tilde\Sigma}\sim N_p(0,24R^{\odot4})\), then \(A_pZ^{\tilde\Sigma}\) has covariance \(\Gamma\), and
\[
    \|A_pZ^{\tilde\Sigma}\|_\infty
    =\max_{1\leq j\leq p}Z_j^{\tilde\Sigma}
    -\min_{1\leq j\leq p}Z_j^{\tilde\Sigma}.
\]
This also proves the equivalent representation and completes the proof. \(\hfill\square\)

\suppappendix{Proofs of Theorem~3.1 and Corollary~3.1}
\label{appendix_c}

\subsection*{Proof of Theorem~3.1}

Put \(d=q_p=p(p-1)/2\), and enumerate \(\mathcal I_p\) in a fixed order when writing a single subscript for a contrast. We first consider an auxiliary multiplier variable without the factor $\hat a_{n,p}$ ,
\[
    \mathring S_n
    =
    \frac1{\sqrt n}
    \sum_{i=1}^n\big(w_i-\bar w\big)e_i.
\]
The bootstrap statistic in the theorem satisfies
\[
    S_n=\hat a_{n,p}\mathring S_n.
\]
Recall from Appendix A that the leading term in the expansion of the pairwise-difference statistic has coordinates
\[
    \mathbb W_{jl}
    =\frac{W_j}{\Sigma_{jj}^2}
    -\frac{W_l}{\Sigma_{ll}^2},
    \qquad (j,l)\in\mathcal I_p,
\]
where
\[
    W_j
    =
    \frac1n\sum_{i=1}^n
    \big\{
    X_{ij}^4-\mathbb EX_{ij}^4
    -
    6r_2\Sigma_{jj}\big(X_{ij}^2-\Sigma_{jj}\big)
    \big\}.
\]
Then \(\operatorname{cov}(w_i)=\Gamma^*\), and
\[
    \sqrt n\,\mathbb W
    =
    \frac1{\sqrt n}\sum_{i=1}^n
    \big(w_i-\mathbb Ew_i\big).
\]

To apply the high-dimensional bootstrap theorem to the absolute maximum, define
\[
    \widetilde w_i
    =
    \big[
    \big(w_i-\mathbb Ew_i\big)^\top,
    -\big(w_i-\mathbb Ew_i\big)^\top
    \big]^\top
    \in\mathbb R^{2d}.
\]
Let
\[
    \widetilde Z\sim N(0,\widetilde\Gamma_p),
    \qquad
    \widetilde\Gamma_p
    =
    \begin{pmatrix}
        \Gamma^* & -\Gamma^*\\
        -\Gamma^* & \Gamma^*
    \end{pmatrix}.
\]
Then
\[
    \max_{1\leq j\leq d}|Z_j^{\Gamma^*}|
    =
    \max_{1\leq j\leq2d}\widetilde Z_j.
\]
The variance lower bound follows from \eqref{nondegenerate_contrasts}. The covariance calculations in Lemma~\ref{lemma_2} and Condition~(A.2) imply
\[
    \max_{1\leq j\leq d}
    \mathbb E\big(w_{ij}-\mathbb Ew_{ij}\big)^2
    =
    O(1).
\]
Moreover, by Lemma~\ref{lemma_elliptical_max_expect},
\[
\begin{aligned}
    \mathbb E
    \Big[\max_{1\leq j\leq d}
    |w_{ij}-\mathbb Ew_{ij}|^4 \Big]
    &\leq
    C
    \mathbb E
    \Big[ \max_{1\leq j\leq p}
    \big(
    |X_{ij}|^{16}+|X_{ij}|^8+1
    \big) \Big]                                                  \\
    &\lesssim
    \log^8p.
\end{aligned}
\]
Thus, taking \(B_n=C\log^2p\) with \(C>0\) sufficiently large gives
\[
    \mathbb E
     \Big[ \max_{1\leq j\leq d}
    \big(
     |w_{ij}-\mathbb Ew_{ij}|/B_n
    \big)^4  \Big]
    \leq2.
\]
The same bound holds for \(\widetilde w_i\), up to a change in the absolute constant. Applying Corollary~4.2 of \cite{victor2017clt} to \(\{\widetilde w_i\}_{i=1}^n\) gives, for every \(\eta\in(0,e^{-1})\), with probability at least \(1-\eta\),
\begin{equation}
\label{eq:auxiliary_bootstrap_bound}
\begin{aligned}
&\sup_{x\in\mathbb R}
\bigg|
\mathbb P\Big(
\max_{1\leq j\leq d}|\mathring S_{nj}|>x
\,\big|\,
\{X_i\}_{i=1}^n
\Big)
-
\mathbb P\Big(
\max_{1\leq j\leq d}|Z_j^{\Gamma^*}|>x
\Big)
\bigg|                                                   \\
&\quad\leq
C\Bigg[
\bigg\{
\frac{B_n^2\log^5(2pn)\log^2(1/\eta)}{n}
\bigg\}^{1/6}
+
\bigg\{
\frac{B_n^2\log^3(2pn)}{\eta^{1/2}n^{1/2}}
\bigg\}^{1/3}
\Bigg].
\end{aligned}
\end{equation}
To deal with $\hat a_{n,p}$, the following lemma shows that the multiplier factor is asymptotically negligible at the rate required by the Gaussian comparison argument. Its proof is placed in Appendix F.
\begin{lemma}
\label{lemma_bootstrap_factor}
Under \(H_0\), assume the conditions of Theorem~3.1. If \(\log p=o(n^{1/14})\), then
\[
    0<\hat a_{n,p}\leq1,
    \qquad
    |\hat a_{n,p}-1|\log^2p=o_{\mathbb P}(1).
\]
\end{lemma}
Conditionally on \(\{X_i\}_{i=1}^n\), the vectors \(\mathring S_n\) and \(S_n\) are centered Gaussian with covariance matrices \(\widetilde\Omega\) and \(\hat a_{n,p}^2\widetilde\Omega\), respectively. Their covariance matrices differ by
\[
    \delta_{n,p}
    =
    |\hat a_{n,p}^2-1|\,\|\widetilde\Omega\|_{\max}.
\]
Because \(\widetilde\Omega\) is positive semidefinite,
\[
    \|\widetilde\Omega\|_{\max}
    \leq
    \max_{1\leq j\leq d}\widetilde\Omega_{jj}
    =
    O_{\mathbb P}(1),
\]
where the last equality follows from \eqref{eq:infeasible_variance_rate}. Lemma~\ref{lemma_bootstrap_factor} therefore yields
\[
    \delta_{n,p}\log^2p=o_{\mathbb P}(1).
\]
By \eqref{eq:infeasible_variance_rate} and
\eqref{nondegenerate_contrasts}, there are constants \(c,C>0\) such that
\[
    \mathcal H_{n,p}
    =
    \bigg\{
    c\leq
    \min_{1\leq j\leq d}\widetilde\Omega_{jj}
    \leq
    \max_{1\leq j\leq d}\widetilde\Omega_{jj}
    \leq C,
    \quad
    \hat a_{n,p}\geq\frac12
    \bigg\}
\]
satisfies \(\mathbb P(\mathcal H_{n,p})\to1\) by
Lemma~\ref{lemma_bootstrap_factor}. On this event, the Gaussian
comparison inequality applied to the corresponding augmented
\(2d\)-dimensional vectors gives
\begin{equation}
\label{eq:factor_comparison}
\begin{aligned}
&1_{\mathcal H_{n,p}}
\sup_{x\in\mathbb R}
\bigg|
\mathbb P\Big(
\max_{1\leq j\leq d}|S_{nj}|>x
\,\big|\,
\{X_i\}_{i=1}^n
\Big)
-
\mathbb P\Big(
\max_{1\leq j\leq d}|\mathring S_{nj}|>x
\,\big|\,
\{X_i\}_{i=1}^n
\Big)
\bigg|                                                   \\
&\quad\leq
C\big(\delta_{n,p}\big)^{1/3}
\bigg\{
1\vee
\log\bigg(
\frac{2d}{\delta_{n,p}}
\bigg)
\bigg\}^{2/3}
=
o_{\mathbb P}(1).
\end{aligned}
\end{equation}
Since \(\mathbb P(\mathcal H_{n,p}^c)\to0\), the supremum in
\eqref{eq:factor_comparison} is \(o_{\mathbb P}(1)\) without the indicator.
Substituting \(B_n=C\log^2p\) in
\eqref{eq:auxiliary_bootstrap_bound} proves the first conclusion of
Theorem~3.1. The comparison in \eqref{eq:factor_comparison} shows that the factor changes the conditional law of the auxiliary bootstrap maximum by \(o_{\mathbb P}(1)\).

We next consider the bootstrap statistic. Define the auxiliary feasible variable
\[
    \mathring{\hat S}_n
    =
    \frac1{\sqrt n}
    \sum_{i=1}^n\big(v_i-\bar v\big)e_i,
\]
so that \(\hat S_n=\hat a_{n,p}\mathring{\hat S}_n\). The same Rosenthal-type argument used in Appendix A gives
\[
    \max_{1\le j\le p}|\hat\Sigma_{jj}-\Sigma_{jj}|=O_{\mathbb P}\bigg(\frac{\log^{3/2}(2p)}{\sqrt{n}}\bigg).
\]
The decomposition in Appendix A also gives
\[
    K_n:=\max_{1\leq j\leq p}|\hat\kappa_j-3r_2|
    =O_{\mathbb P}\bigg(\frac{\log^{5/2}(2p)}{\sqrt n}\bigg).
\]
In particular, all \(\hat\kappa_j\)'s are bounded with probability tending to one. On this event and
\(\{\max_j|\hat\Sigma_{jj}-\Sigma_{jj}|\leq b/2\}\), a first-order Taylor expansion of the map
\[
    (a,k)\mapsto a^{-2}(x^4-2kax^2)
\]
over a compact neighborhood of \((\Sigma_{jj},3r_2)\), together with Condition~(A.2), yields
\[
\begin{aligned}
    \max_{1\le j\le d}
    \frac1n\sum_{i=1}^n
    |v_{ij}-w_{ij}|^2
    &\le
    C
    \bigg\{
    \max_{1\le j\le p}|\hat\Sigma_{jj}-\Sigma_{jj}|^2
    +
    K_n^2
    \bigg\}                                      \\
    &\quad \times
    \bigg[
    1+
    \frac1n\sum_{i=1}^n
    \max_{1\le j\le p}
    \big(|X_{ij}|^8+|X_{ij}|^4\big)
    \bigg].
\end{aligned}
\]
Let \(D_n=\max_{1\leq j\leq d}n^{-1}\sum_{i=1}^n|v_{ij}-w_{ij}|^2\), as in Appendix~F. Conditionally on the data,
\(\mathring{\hat S}_n-\mathring S_n\) is centered Gaussian and every
coordinate variance is at most \(D_n\). The conditional Gaussian maximal
inequality and \eqref{eq:Dn_rate} therefore give
\begin{equation}
\label{eq:feasible_gaussian_max_rate}
\begin{aligned}
    \mathbb E_*
    \max_{1\leq j\leq d}
    |\mathring{\hat S}_{nj}-\mathring S_{nj}|
    \leq
    C\sqrt{\log(2p)}D_n^{1/2} =
    O_{\mathbb P}\bigg(
    \frac{\log^5(2p)}{\sqrt n}
    \bigg).
\end{aligned}
\end{equation}
Under \(\log p=o(n^{1/14})\), the last display is
\(o_{\mathbb P}\{\log^{-1/2}(2p)\}\). Since
\(0<\hat a_{n,p}\leq1\), the same bound applies to
\(\max_j|\hat S_{nj}-S_{nj}|\). Here \(\mathbb P_*\) and
\(\mathbb E_*\) denote probability and expectation with respect to the
bootstrap multipliers conditional on the data.
For every \(\varepsilon>0\),
\[
\begin{aligned}
&\sup_{x\in\mathbb R}
\bigg|
\mathbb P\Big(
\max_{1\leq j\leq d}|\hat S_{nj}|\leq x
\,\big|\,
\{X_i\}_{i=1}^n
\Big)
-
\mathbb P\Big(
\max_{1\leq j\leq d}|S_{nj}|\leq x
\,\big|\,
\{X_i\}_{i=1}^n
\Big)
\bigg|                                                   \\
&\quad\leq
\mathbb P\Big(
\max_{1\leq j\leq d}|\hat S_{nj}-S_{nj}|>\varepsilon
\,\big|\,
\{X_i\}_{i=1}^n
\Big)                                                   \\
&\qquad+
\sup_{x\in\mathbb R}
\mathbb P\Big(
 x-\varepsilon
 <
 \max_{1\leq j\leq d}|S_{nj}|
 \leq
 x+\varepsilon
\,\big|\,
\{X_i\}_{i=1}^n
\Big).
\end{aligned}
\]
By \eqref{eq:feasible_gaussian_max_rate}, there exists a deterministic
sequence \(\zeta_{n,p}\downarrow0\) such that
\(\zeta_{n,p}\sqrt{\log(2p)}\to0\) and, upon taking
\(\varepsilon=\zeta_{n,p}\), the first term is
\(o_{\mathbb P}(1)\). By Lemma~\ref{lemma_bootstrap_factor},
\(\hat a_{n,p}\to1\) in probability, so the conditional coordinate
variances of \(S_n\) remain uniformly bounded away from zero with probability
tending to one. The Gaussian anti-concentration inequality then makes the
second term \(o_{\mathbb P}(1)\), and hence
\begin{equation}
\label{eq:feasible_infeasible_comparison}
\sup_{x\in\mathbb R}
\bigg|
\mathbb P\Big(
\max_{1\leq j\leq d}|\hat S_{nj}|>x
\,\big|\,
\{X_i\}_{i=1}^n
\Big)
-
\mathbb P\Big(
\max_{1\leq j\leq d}|S_{nj}|>x
\,\big|\,
\{X_i\}_{i=1}^n
\Big)
\bigg|
=
 o_{\mathbb P}(1).
\end{equation}

Finally, Appendix B established that
\[
    \max_{1\leq j,k\leq q_p}
    |\Gamma^*_{jk}-\Gamma_{jk}|
    \log^2p
    \to0.
\]
The Gaussian comparison inequality implies
\[
\sup_{x\in\mathbb R}
\bigg|
\mathbb P\Big(
\max_{1\leq j\leq d}|Z_j^{\Gamma^*}|>x
\Big)
-
\mathbb P\Big(
\max_{1\leq j\leq d}|Z_j^{\Gamma}|>x
\Big)
\bigg|
=
 o(1).
\]
Combining this display with the first conclusion of Theorem~3.1 and \eqref{eq:feasible_infeasible_comparison} proves the second conclusion of Theorem~3.1. This completes the proof. \(\hfill\square\)

\subsection*{Proof of Corollary~3.1}

By Theorem~2.1,
\[
\begin{aligned}
&\sup_{x\in\mathbb R}
\bigg|
\mathbb P(T_n>x)
-
\mathbb P\Big(
\max_{1\leq j\leq q_p}|Z_j^{\Gamma^*}|>x
\Big)
\bigg|                                                   \\
&\quad\leq
C\varepsilon_{n,p}.
\end{aligned}
\]
On the other hand, Theorem~3.1 gives, with probability at least \(1-\eta\),
\[
\begin{aligned}
&\sup_{x\in\mathbb R}
\bigg|
\mathbb P\Big(
\max_{1\leq j\leq q_p}|\mathring S_{nj}|>x
\,\big|\,
\{X_i\}_{i=1}^n
\Big)
-
\mathbb P\Big(
\max_{1\leq j\leq q_p}|Z_j^{\Gamma^*}|>x
\Big)
\bigg|                                                   \\
&\quad\leq
C\Bigg[
\bigg\{
\frac{\log^9(2pn)\log^2(1/\eta)}{n}
\bigg\}^{1/6}
+
\bigg\{
\frac{\log^7(2pn)}{\eta^{1/2}n^{1/2}}
\bigg\}^{1/3}
\Bigg]
.
\end{aligned}
\]
The triangle inequality proves the second assertion of Corollary~3.1.
The second conclusion of Theorem~3.1 and Proposition~2.1 then yield
\[
\sup_{x\in\mathbb R}
\Big|
\mathbb P\big(T_n>x\big)
-
\mathbb P\Big(
\hat M_n>x
\,\big|\,
\{X_i\}_{i=1}^n
\Big)
\Big|
=
 o_{\mathbb P}(1),
\]
which proves the first assertion of Corollary~3.1.

It remains to justify the random critical value. Define
\[
    F_n(x)=\mathbb P(T_n\leq x),
    \qquad
    F_n^G(x)
    =
    \mathbb P\bigg(
    \max_{1\leq j\leq q_p}
    |Z_j^{\Gamma}|
    \leq x
    \bigg),
\]
and
\[
    \hat F_n(x)
    =
    \mathbb P\big(
    \hat M_n\leq x
    \,\big|\,
    \{X_i\}_{i=1}^n
    \big).
\]
The preceding results imply
\[
    \rho_n
    :=
    \sup_x|F_n(x)-F_n^G(x)|
    \to0,
    \qquad
    \beta_n
    :=
    \sup_x|\hat F_n(x)-F_n^G(x)|
    =
    o_{\mathbb P}(1).
\]
Choose a deterministic sequence \(\gamma_n\downarrow0\) such that
\(\rho_n\leq\gamma_n\) for all sufficiently large \(n\) and
\(\mathbb P(\beta_n>\gamma_n)\to0\). Let
\[
    q_n^G(u)
    =
    \inf\big\{
    x:F_n^G(x)\geq u
    \big\}.
\]
Condition~(A.2) implies that the marginal variances of the Gaussian
contrasts are bounded away from zero. Hence \(F_n^G\) is continuous, and
on the event \(\{\beta_n\leq\gamma_n\}\),
\[
    q_n^G(1-\alpha-\gamma_n)
    \leq
    \hat c_{1-\alpha}
    \leq
    q_n^G(1-\alpha+\gamma_n).
\]
It follows that
\[
\begin{aligned}
    \alpha-\gamma_n-\rho_n
    -\mathbb P(\beta_n>\gamma_n)
    &\leq
    \mathbb P_{H_0}\big(
    T_n>\hat c_{1-\alpha}
    \big)
    \\
    &\leq
    \alpha+\gamma_n+\rho_n
    +\mathbb P(\beta_n>\gamma_n).
\end{aligned}
\]
Both bounds converge to \(\alpha\), and therefore
\[
    \mathbb P_{H_0}\big(
    T_n>\hat c_{1-\alpha}
    \big)
    \to\alpha.
\]
Thus the proposed bootstrap test has asymptotic size $\alpha$, which completes the proof. \(\hfill\square\)

\suppappendix{Proof of Proposition~3.1}
\label{appendix_d}
Let
\[
    \mu_{2j}
    =
    \mathbb E(X_{1j}^2)
    =
    \Sigma_{jj},
    \qquad
    \mu_{4j}
    =
    \mathbb E(X_{1j}^4)
    =
    \kappa_j\mu_{2j}^2,
\]
and define their empirical counterparts by
\[
    \hat\mu_{2j}
    =
    \frac1n\sum_{i=1}^nX_{ij}^2,
    \qquad
    \hat\mu_{4j}
    =
    \frac1n\sum_{i=1}^nX_{ij}^4.
\]
Then
\[
    \hat\kappa_j
    =
    \frac{\hat\mu_{4j}}{\hat\mu_{2j}^2},
    \qquad
    \kappa_j
    =
    \frac{\mu_{4j}}{\mu_{2j}^2}.
\]

We divide the proof into two steps. In the first step, we control the empirical second and fourth moments. Applying the maximal inequality in \eqref{eq:maximal_sample_mean} to
\[
    Z_{ij}^{(4)}
    =
    X_{ij}^4-\mathbb E(X_{1j}^4),
    \qquad
    1\leq j\leq p,
\]
gives
\[
\begin{aligned}
\mathbb E \Big[
\max_{1\leq j\leq p}
|\hat\mu_{4j}-\mu_{4j}|  \Big]
&\leq
C
\sqrt{\frac{\log(2p)}{n}}
\bigg\{
\mathbb E  \Big[ 
\max_{1\leq j\leq p}
\big|
X_{1j}^4-\mathbb E(X_{1j}^4)
\big|^2 \Big]
\bigg\}^{1/2}.
\end{aligned}
\]
By Jensen's inequality,
\[
    \mathbb E
    \Big[ \max_{1\leq j\leq p}
    X_{1j}^8 \Big]
    \leq
    \bigg\{
    \mathbb E
    \Big[ \max_{1\leq j\leq p}
    X_{1j}^{16} \Big]
    \bigg\}^{1/2}
    =
    M_{16,n,p}^{1/2},
\]
and
\[
    \max_{1\leq j\leq p}
    \big\{
    \mathbb E(X_{1j}^4)
    \big\}^2
    \leq
    M_{16,n,p}^{1/2}.
\]
Consequently,
\begin{equation}
\label{eq:power_fourth_moment_rate}
    \max_{1\leq j\leq p}
    |\hat\mu_{4j}-\mu_{4j}|
    =
    O_{\mathbb P}\bigg(
    M_{16,n,p}^{1/4}
    \sqrt{\frac{\log p}{n}}
    \bigg).
\end{equation}
Similarly, applying \eqref{eq:maximal_sample_mean} to
\[
    Z_{ij}^{(2)}
    =
    X_{ij}^2-\mathbb E(X_{1j}^2), \quad 1\leq j \leq p,
\]
and using \(
    \mathbb E
    (\max_{1\leq j\leq p}
    X_{1j}^4)
    \leq
    M_{16,n,p}^{1/4}
\), we obtain
\begin{equation}
\label{eq:power_second_moment_rate}
    \max_{1\leq j\leq p}
    |\hat\mu_{2j}-\mu_{2j}|
    =
    O_{\mathbb P}\bigg(
    M_{16,n,p}^{1/8}
    \sqrt{\frac{\log p}{n}}
    \bigg).
\end{equation}
Condition~(C.1) includes \(\kappa_j\leq B\), so we have $\Delta\leq B$. Condition~(C.3) therefore implies
\[
    M_{16,n,p}^{1/4}
    \sqrt{\frac{\log p}{n}}
    =
    o(1).
\]
Since \(\Sigma_{jj}\geq b\), Jensen's inequality gives
\[
    M_{16,n,p}
    \geq
    \max_{1\leq j\leq p}
    \mathbb E(X_{1j}^{16})
    \geq
    b^8.
\]
It follows that \( \sqrt{\frac{\log p}{n}} = o(1)\) and hence
\[
\begin{aligned}
    M_{16,n,p}^{1/8}
    \sqrt{\frac{\log p}{n}}
    &=
    \bigg\{
    M_{16,n,p}^{1/4}
    \sqrt{\frac{\log p}{n}}
    \bigg\}^{1/2}
    \bigg(
    \frac{\log p}{n}
    \bigg)^{1/4} = o(1).
\end{aligned}
\]
Thus, by \eqref{eq:power_second_moment_rate}, the event
\[
    \mathcal E_n
    =
    \bigg\{
    \max_{1\leq j\leq p}
    |\hat\mu_{2j}-\mu_{2j}|
    \leq
    \frac b2
    \bigg\}
\]
satisfies
\(    \mathbb P(\mathcal E_n)\to1
\). On \(\mathcal E_n\), we notice
\(
    \hat\mu_{2j}
    \geq
    \frac{\mu_{2j}}2
    \geq
    \frac b2
\), and
\[
\begin{aligned}
|\hat\kappa_j-\kappa_j|
&\leq
\frac{|\hat\mu_{4j}-\mu_{4j}|}{\hat\mu_{2j}^2}
+
\mu_{4j}
\bigg|
\frac1{\hat\mu_{2j}^2}
-
\frac1{\mu_{2j}^2}
\bigg|                                                \\
&\leq
\frac4{b^2}
|\hat\mu_{4j}-\mu_{4j}|
+
\kappa_j
\frac{
|\hat\mu_{2j}-\mu_{2j}|
\big(\hat\mu_{2j}+\mu_{2j}\big)
}{
\hat\mu_{2j}^2
}                                                     \\
&\leq
C
\big\{
|\hat\mu_{4j}-\mu_{4j}|
+
|\hat\mu_{2j}-\mu_{2j}|
\big\},
\end{aligned}
\]
where \(C>0\) depends only on \(b\) and \(B\). Combining
\eqref{eq:power_fourth_moment_rate} and
\eqref{eq:power_second_moment_rate} yields
\begin{equation}
\label{eq:power_kurtosis_rate}
    \sqrt n
    \max_{1\leq j\leq p}
    |\hat\kappa_j-\kappa_j|
    =
    O_{\mathbb P}\big(
    M_{16,n,p}^{1/4}\sqrt{\log p}
    \big).
\end{equation}
It follows that
\begin{equation}
\label{eq:power_pairwise_rate}
\begin{aligned}
&\sqrt n
\max_{1\leq j<l\leq p}
\big|
\hat\kappa_j-\hat\kappa_l
-
\big(\kappa_j-\kappa_l\big)
\big|                                                 \\
&\quad\leq
2\sqrt n
\max_{1\leq j\leq p}
|\hat\kappa_j-\kappa_j|                              \\
&\quad=
O_{\mathbb P}\big(
M_{16,n,p}^{1/4}\sqrt{\log p}
\big).
\end{aligned}
\end{equation}
In the second step, we establish upper bound for the bootstrap critical value. Define the multiplier bootstrap variable without the factor \(\hat a_{n,p}\) by
\[
    \mathring{\hat S}_n
    =
    \frac1{\sqrt n}
    \sum_{i=1}^n
    \big(v_i-\bar v\big)e_i,
\]
and let \(\mathring c_{1-\alpha}\) denote the conditional
\((1-\alpha)\)-quantile of
\(\|\mathring{\hat S}_n\|_\infty\). By the definition of the proposed
bootstrap statistic, it satisfies  \(
    \hat S_n
    =
    \hat a_{n,p}\mathring{\hat S}_n
\). By the definition of $\hat a_{n,p}$, we know \(0<\hat a_{n,p}\leq 1\). Consequently,
\(
    \|\hat S_n\|_\infty
    \leq
    \|\mathring{\hat S}_n\|_\infty
\) conditionally on the observations, and hence
\begin{equation}
\label{eq:power_quantile_comparison}
    \hat c_{1-\alpha}
    \leq
    \mathring c_{1-\alpha}.
\end{equation}
In particular, no convergence of \(\hat a_{n,p}\) to one is required
under \(H_1\). Recall that, conditionally on the observations,
\( \mathring{\hat S}_n \sim N_{q_p}(0,\hat\Omega)\). By \eqref{eq:power_kurtosis_rate} and Conditions~(C.1)--(C.3),
\[
    \max_{1\leq j\leq p}\hat\kappa_j\leq B+o_{\mathbb P}(1).
\]
We henceforth intersect \(\mathcal E_n\) with \(\{\max_j\hat\kappa_j\leq B+1\}\), which still has probability tending to one. On this event, all \(\hat\Sigma_{jj}\)'s are bounded
away from zero. It follows from the definition of \(v_{i,jl}\) that
\[
    |v_{i,jl}|
    \leq C\big(1+|X_{ij}|^4+|X_{il}|^4\big),
    \qquad (j,l)\in\mathcal I_p,
\]
where \(C>0\) depends only on \(b\) and \(B\). Therefore,
\[
    \max_{1\leq j\leq q_p}
    v_{ij}^2
    \leq
    C
    \Big(
    1+
    \max_{1\leq k\leq p}|X_{ik}|^8
    \Big).
\]
Since the empirical variance is bounded by the empirical second
moment, we obtain
\[
\begin{aligned}
    \max_{1\leq j\leq q_p}\hat\Omega_{jj}
    &\leq
    \frac1n\sum_{i=1}^n
    \max_{1\leq j\leq q_p}v_{ij}^2                         
    \leq
    \frac Cn\sum_{i=1}^n
    \bigg(
    1+
    \max_{1\leq k\leq p}|X_{ik}|^8
    \bigg).
\end{aligned}
\]
Since \(M_{16,n,p}\geq b^8\) and \(
    \mathbb E (
    \max_{1\leq k\leq p}|X_{1k}|^8 )
    \leq
    M_{16,n,p}^{1/2}.
\), Markov's inequality gives
\begin{equation}
\label{eq:power_bootstrap_variance}
    \max_{1\leq j\leq q_p}\hat\Omega_{jj}
    =
    O_{\mathbb P}\big(
    M_{16,n,p}^{1/2}
    \big).
\end{equation}
Put \(
    \hat\sigma_{\max}^2
    =
    \max_{1\leq j\leq q_p}\hat\Omega_{jj}
\). If \(\hat\sigma_{\max}=0\), then \(\mathring c_{1-\alpha}=0\) and the required upper bound is immediate. Otherwise, conditionally on the data, the Gaussian union bound yields
\[
\begin{aligned}
\mathbb P\Big(
\|\mathring{\hat S}_n\|_\infty>t
\,\big|\,
\{X_i\}_{i=1}^n
\Big)
&\leq
2q_p
\exp\bigg(
-\frac{t^2}{2\hat\sigma_{\max}^2}
\bigg).
\end{aligned}
\]
Therefore,
\[
    \mathring c_{1-\alpha}
    \leq
    \hat\sigma_{\max}
    \bigg\{
    2\log\bigg(
    \frac{2q_p}{\alpha}
    \bigg)
    \bigg\}^{1/2}.
\]
Combining this bound with \eqref{eq:power_quantile_comparison} and
\eqref{eq:power_bootstrap_variance}, we obtain
\begin{equation}
\label{eq:power_critical_value}
    \hat c_{1-\alpha}
    =
    O_{\mathbb P}\big(
    M_{16,n,p}^{1/4}\sqrt{\log p}
    \big).
\end{equation}
Finally, by the definition of \(\Delta\),
\[
\begin{aligned}
T_n
&=
\sqrt n
\max_{1\leq j<l\leq p}
|\hat\kappa_j-\hat\kappa_l|                         \\
&\geq
\sqrt n\,\Delta
-
\sqrt n
\max_{1\leq j<l\leq p}
\big|
\hat\kappa_j-\hat\kappa_l
-
\big(\kappa_j-\kappa_l\big)
\big|.
\end{aligned}
\]
It follows from \eqref{eq:power_pairwise_rate} that
\[
    T_n
    \geq
    \sqrt n\,\Delta
    -
    O_{\mathbb P}\big(
    M_{16,n,p}^{1/4}\sqrt{\log p}
    \big).
\]
Together with \eqref{eq:power_critical_value}, this gives
\[
\begin{aligned}
T_n-\hat c_{1-\alpha}
&\geq
\sqrt n\,\Delta
-
O_{\mathbb P}\big(
M_{16,n,p}^{1/4}\sqrt{\log p}
\big)                                                 =
\sqrt n\,\Delta
\big\{
1-o_{\mathbb P}(1)
\big\},
\end{aligned}
\]
where the last equality follows from Condition~(C.3). Consequently,
\[
    \mathbb P_{H_1}\big(
    T_n>\hat c_{1-\alpha}
    \big)
    \to 1.
\]
This completes the proof. \(\hfill\square\)

\suppappendix{Proof of Lemma~\ref{lemma_elliptical_max_expect}}
\label{appendix_e}

Write \(A=\Sigma^{1/2}\), let \(a_j^\top\) be its \(j\)th row, and set
\(C_{\Sigma}=\max_j\Sigma_{jj}=O(1)\) and
\(W=\max_{1\leq j\leq p}|\langle a_j,u_1\rangle|\).
Since \(\|a_j\|^2=\Sigma_{jj}\leq C_{\Sigma}\), the spherical inner-product tail bound and the union bound give
\[
    \mathbb P(W\geq t)
    \leq\min\big\{1,2p\exp(-ct^2)\big\},
    \qquad c=\frac{p}{2C_{\Sigma}},\quad t\geq0.
\]
Put \(x_0=\log(2p)\) and \(t_0=(x_0/c)^{1/2}\). For each fixed \(k\geq1\), integrating the tail and substituting \(y=ct^2\) yield
\[
\begin{aligned}
    \mathbb EW^{2k}
    &=2k\int_0^\infty t^{2k-1}\mathbb P(W\geq t)\,dt\\
    &\leq t_0^{2k}+4kp\int_{t_0}^\infty t^{2k-1}e^{-ct^2}\,dt\\
    &=c^{-k}\bigg\{x_0^k+2kp\int_{x_0}^\infty y^{k-1}e^{-y}\,dy\bigg\}
    =O\bigg(\frac{\log^k p}{p^k}\bigg).
\end{aligned}
\]
The last step uses
\(\int_x^\infty y^{k-1}e^{-y}\,dy=O(e^{-x}x^{k-1})\) as \(x\to\infty\).
Finally, independence of \(\xi_1\) and \(u_1\), together with
\(\mathbb E|\xi_1|^{2k}=O(p^k)\), gives
\[
    \mathbb E\Big[\max_{1\leq j\leq p}X_{1j}^{2k}\Big]
    =\mathbb E|\xi_1|^{2k}\,\mathbb EW^{2k}
    =O(\log^k p).
\]
This completes the proof. \(\hfill\square\)

\suppappendix{Proof of Lemma~\ref{lemma_bootstrap_factor}}
\label{appendix_f}
Let \(d=q_p=p(p-1)/2\), and recall the definitions
\[
    \widetilde\Omega
    =
    \frac1n\sum_{i=1}^n
    \big(w_i-\bar w\big)\big(w_i-\bar w\big)^\top,
    \qquad
    \hat\Omega
    =
    \frac1n\sum_{i=1}^n
    \big(v_i-\bar v\big)\big(v_i-\bar v\big)^\top.
\]
We first record a maximal inequality that will be used repeatedly. Let \(Z_1,\ldots,Z_n\) be i.i.d. mean-zero random vectors in \(\mathbb R^m\). By symmetrization and the conditional Rademacher maximal inequality,
\begin{equation}
\label{eq:maximal_sample_mean}
    \mathbb E \Bigg[
    \max_{1\leq j\leq m}
    \bigg|
    \frac1n\sum_{i=1}^n Z_{ij}
    \bigg|  \Bigg]
    \leq
    C
    \sqrt{\frac{\log(2m)}{n}}
    \bigg\{
    \mathbb E \bigg[
    \max_{1\leq j\leq m}Z_{1j}^2 \bigg]
    \bigg\}^{1/2}.
\end{equation}
Indeed, if \(\varepsilon_1,\ldots,\varepsilon_n\) are independent Rademacher variables, then
\[
\begin{aligned}
    \mathbb E \Bigg[ 
    \max_{1\leq j\leq m}
    \bigg|
    \frac1n\sum_{i=1}^n Z_{ij}
    \bigg| \Bigg]
    &\leq
    \frac2n
    \mathbb E \bigg\{ 
    \mathbb E_{\varepsilon}
    \max_{1\leq j\leq m}
    \bigg|
    \sum_{i=1}^n\varepsilon_iZ_{ij}
    \bigg|  \bigg\}                                                  \\
    &\leq
    \frac{C\sqrt{\log(2m)}}{n}
    \mathbb E \bigg[ 
    \max_{1\leq j\leq m}
    \bigg(
    \sum_{i=1}^n Z_{ij}^2
    \bigg)^{1/2}   \bigg]                                          \\
    &\leq
    C
    \sqrt{\frac{\log(2m)}{n}}
    \bigg\{
    \mathbb E
    \bigg[\max_{1\leq j\leq m}Z_{1j}^2 \bigg]
    \bigg\}^{1/2}.
\end{aligned}
\]

Recall that each \(w_{ij}\) is a linear combination of normalized fourth- and second-power terms. By Condition~(A.2), Lemma~\ref{lemma_elliptical_max_expect}, and \(\mathbb E\xi_1^{2k}=O(p^k)\) for \(k\leq8\),
\begin{equation}
\label{eq:w_max_moments}
    \mathbb E \Big[
    \max_{1\leq j\leq d}
    |w_{ij}-\mathbb Ew_{ij}|^2 \Big]
    =
    O(\log^4p),
    \qquad
    \mathbb E \Big[
    \max_{1\leq j\leq d}
    |w_{ij}-\mathbb Ew_{ij}|^4 \Big]
    =
    O(\log^8p).
\end{equation}
Apply \eqref{eq:maximal_sample_mean} to
\[
    Z_{ij}
    =
    \big(w_{ij}-\mathbb Ew_{ij}\big)^2
    -
    \mathbb E\big(w_{ij}-\mathbb Ew_{ij}\big)^2,
\]
together with \eqref{eq:w_max_moments}, this gives
\begin{equation}
\label{eq:infeasible_variance_rate}
\begin{aligned}
    \max_{1\leq j\leq d}
    \big|
    \widetilde\Omega_{jj}-\Gamma^*_{jj}
    \big|
    &\leq
    \max_{1\leq j\leq d}
    \bigg|
    \frac1n\sum_{i=1}^n
    \bigg[
    \big(w_{ij}-\mathbb Ew_{ij}\big)^2
    -
    \mathbb E\big(w_{ij}-\mathbb Ew_{ij}\big)^2
    \bigg]
    \bigg|                                                   \\
    &\quad+
    \max_{1\leq j\leq d}
    |\bar w_j-\mathbb Ew_{ij}|^2                            \\
    &=
    O_{\mathbb P}\bigg(
    \frac{\log^{9/2}p}{\sqrt n}
    \bigg),
\end{aligned}
\end{equation}
where $\bar w_j = n^{-1}\sum_{i=1}^n w_{ij}$. Here the sample-mean term is
\(O_{\mathbb P}(\log^5p/n)\), and is absorbed by the displayed rate under \(\log p=o(n^{1/14})\). We next compare the feasible and infeasible influence vectors,  $v_i$ and $w_i$. Put
\[
    D_n
    =
    \max_{1\leq j\leq d}
    \frac1n\sum_{i=1}^n|v_{ij}-w_{ij}|^2,
    \qquad
    \Delta_{\Sigma,n}
    =
    \max_{1\leq j\leq p}
    |\hat\Sigma_{jj}-\Sigma_{jj}|.
\]
Put \(K_n=\max_j|\hat\kappa_j-3r_2|\). On the event \(\{\Delta_{\Sigma,n}\leq b/2,K_n\leq1\}\), whose probability tends to one, \(D_n\) satisfies
\begin{equation}
\label{eq:plug_in_L2}
\begin{aligned}
    D_n
    &\leq
    C\big(
    \Delta_{\Sigma,n}^2+K_n^2
    \big)                                             
    \bigg[
    1+
    \frac1n\sum_{i=1}^n
    \max_{1\leq j\leq p}
    \big(
    |X_{ij}|^8+|X_{ij}|^4
    \big)
    \bigg].
\end{aligned}
\end{equation}
Applying \eqref{eq:maximal_sample_mean} to \(X_{ij}^2-\mathbb EX_{ij}^2\), and using Lemma~\ref{lemma_elliptical_max_expect}, yields
\begin{equation}
\label{eq:sample_variance_rate}
    \Delta_{\Sigma,n}
    =
    O_{\mathbb P}\bigg(
    \frac{\log^{3/2}p}{\sqrt n}
    \bigg),
    \qquad
    \frac1n\sum_{i=1}^n
    \max_{1\leq j\leq p}
    \big(
    |X_{ij}|^8+|X_{ij}|^4
    \big)
    =
    O_{\mathbb P}(\log^4p).
\end{equation}
By \eqref{decomp_first}, on \(\{\Delta_{\Sigma,n}\leq b/2\}\),
\[
    K_n\leq C\bigg(\max_j|W_j|+\Delta_{\Sigma,n}^2\bigg).
\]
The maximal bounds in \eqref{eq:appendix_a_maximal_rates} and \eqref{eq:sample_variance_rate} therefore imply
\[
    K_n=O_{\mathbb P}\bigg(\frac{\log^{5/2}p}{\sqrt n}\bigg)=o_{\mathbb P}(1).
\]
It follows from \eqref{eq:plug_in_L2}--\eqref{eq:sample_variance_rate} that
\begin{equation}
\label{eq:Dn_rate}
    D_n^{1/2}
    =O_{\mathbb P}\bigg(\frac{\log^{9/2}p}{\sqrt n}\bigg).
\end{equation}
By the Cauchy--Schwarz inequality,
\[
\begin{aligned}
    \max_{1\leq j\leq d}
    |\hat\Omega_{jj}-\widetilde\Omega_{jj}|
    &\leq
    D_n
    +
    2D_n^{1/2}
    \max_{1\leq j\leq d}
    \widetilde\Omega_{jj}^{1/2}.
\end{aligned}
\]
The covariance calculations in Appendix A imply \(\max_j\Gamma^*_{jj}=O(1)\), and hence \eqref{eq:infeasible_variance_rate} gives \(\max_j\widetilde\Omega_{jj}=O_{\mathbb P}(1)\). Therefore,
\begin{equation}
\label{eq:feasible_variance_rate}
    \max_{1\leq j\leq d}
    |\hat\Omega_{jj}-\Gamma^*_{jj}|
    =
    O_{\mathbb P}\bigg(
    \frac{\log^{9/2}p}{\sqrt n}
    \bigg).
\end{equation}
We next apply these bounds to the centered coordinate scores used in the correction factor. Recall
\[
    h_{ij}=\Sigma_{jj}^{-2}\big(X_{ij}^4-6r_2\Sigma_{jj}X_{ij}^2\big),
    \qquad
    \hat h_{ij}=\hat\Sigma_{jj}^{-2}\big(X_{ij}^4-2\hat\kappa_j\hat\Sigma_{jj}X_{ij}^2\big),
\]
and define
\[
    g_{ij}=h_{ij}-\frac1p\sum_{k=1}^p h_{ik},
    \qquad
    \hat g_{ij}=\hat h_{ij}-\frac1p\sum_{k=1}^p\hat h_{ik}.
\]
Let \(C_p=I_p-p^{-1}\mathbf1_p\mathbf1_p^\top\), so that \(g_i=C_ph_i\). Every row of \(C_p\) has \(\ell_1\)-norm at most two. Thus the moment bounds in \eqref{eq:w_max_moments} also hold for \(g_{ij}\), with the maximum over \(1\leq j\leq p\). By the covariance calculation in Appendix A,
\[
    \sigma_{gj}^2:=\operatorname{var}(g_{ij})
    =\big(C_p\tilde\Sigma^*C_p\big)_{jj}.
\]
Its limiting counterpart is
\[
    \omega_j
    =24\big(C_pR^{\odot4}C_p\big)_{jj}
    =24\bigg(1-\frac2p\sum_{k=1}^p r_{jk}^4
    +\frac1{p^2}\sum_{k,m=1}^p r_{km}^4\bigg),
\]
and the row-norm bound together with Lemma~\ref{lemma_1} gives
\begin{equation}
\label{eq:population_variance_limit}
    \max_{1\leq j\leq p}|\sigma_{gj}^2-\omega_j|
    \leq4\|\tilde\Sigma^*-24R^{\odot4}\|_{\max}
    =o(p^{-3/4}).
\end{equation}
We first verify that these variances are uniformly nondegenerate. Let \(Y\sim N_p(0,R^{\odot4})\), and put \(\bar Y=p^{-1}\sum_kY_k\). Since \(\operatorname{var}(Y_j)=1\), the Cauchy--Schwarz inequality gives
\[
\begin{aligned}
    \frac{\omega_j}{24}
    =\operatorname{var}(Y_j-\bar Y)
    &\geq\big\{\operatorname{cov}(Y_j,Y_j-\bar Y)\big\}^2\\
    &=\bigg(1-\frac1p\sum_{k=1}^p r_{jk}^4\bigg)^2
    \geq\bigg\{\frac{p-1}{p}(1-r_0^4)\bigg\}^2.
\end{aligned}
\]
Therefore \(\inf_j\omega_j\geq6(1-r_0^4)^2>0\) for \(p\geq2\), and also \(\sup_j\omega_j\leq96\) by the row-norm bound. 

Write \(\bar g_j=n^{-1}\sum_i g_{ij}\) and \(\bar{\hat g}_j=n^{-1}\sum_i\hat g_{ij}\). Applying \eqref{eq:maximal_sample_mean} to the centered squared scores, as in \eqref{eq:infeasible_variance_rate}, yields
\[
    \max_{1\leq j\leq p}\bigg|
    \frac1n\sum_{i=1}^n(g_{ij}-\bar g_j)^2-\sigma_{gj}^2
    \bigg|
    =O_{\mathbb P}\bigg(\frac{\log^{9/2}p}{\sqrt n}\bigg).
\]
The same Taylor expansion as in \eqref{eq:plug_in_L2} yields
\[
    D_n^h:=\max_{1\leq j\leq p}\frac1n\sum_{i=1}^n|\hat h_{ij}-h_{ij}|^2,
    \qquad
    (D_n^h)^{1/2}
    =O_{\mathbb P}\bigg(\frac{\log^{9/2}p}{\sqrt n}\bigg).
\]
By Jensen's inequality,
\[
    D_n^g:=\max_{1\leq j\leq p}\frac1n\sum_{i=1}^n|\hat g_{ij}-g_{ij}|^2
    \leq4D_n^h.
\]
Centering over the observations contracts the empirical \(L_2\) norm. Hence, by the Cauchy--Schwarz inequality,
\[
\begin{aligned}
&\max_{1\leq j\leq p}\bigg|
\frac1n\sum_{i=1}^n(\hat g_{ij}-\bar{\hat g}_j)^2
-\frac1n\sum_{i=1}^n(g_{ij}-\bar g_j)^2\bigg|\\
&\quad\leq D_n^g+2(D_n^g)^{1/2}
\max_{1\leq j\leq p}\bigg\{\frac1n\sum_{i=1}^n(g_{ij}-\bar g_j)^2\bigg\}^{1/2}.
\end{aligned}
\]
The maximum on the right is \(O_{\mathbb P}(1)\), and consequently
\[
    \max_{1\leq j\leq p}\bigg|
    \frac1n\sum_{i=1}^n(\hat g_{ij}-\bar{\hat g}_j)^2-\omega_j
    \bigg|
    =O_{\mathbb P}\bigg(\frac{\log^{9/2}p}{\sqrt n}\bigg)
    +o(p^{-3/4}).
\]
To control the estimated reference variances, apply \eqref{eq:maximal_sample_mean} to the \(p^2\) products \(X_{ij}X_{ik}-\mathbb E(X_{ij}X_{ik})\). Since \(\max_{j,k}|X_{ij}X_{ik}|^2\leq\max_j|X_{ij}|^4\), this gives
\[
    \max_{1\leq j,k\leq p}|\hat\Sigma_{jk}-\Sigma_{jk}|
    =O_{\mathbb P}\bigg(\frac{\log^{3/2}p}{\sqrt n}\bigg).
\]
Together with \eqref{eq:sample_variance_rate} and Condition~(A.2), a Taylor expansion gives
\begin{equation}
\label{eq:correlation_rate}
    \max_{1\leq j,k\leq p}|\hat r_{jk}-r_{jk}|
    =O_{\mathbb P}\bigg(\frac{\log^{3/2}p}{\sqrt n}\bigg).
\end{equation}
Recall
\[
    \hat\omega_j
    =24\bigg(1-\frac2p\sum_{k=1}^p\hat r_{jk}^4
    +\frac1{p^2}\sum_{k,m=1}^p\hat r_{km}^4\bigg).
\]
The map \(x\mapsto x^4\) is Lipschitz on \([-1,1]\). Therefore,
\begin{equation}
\label{eq:omega_rate}
    \max_{1\leq j\leq p}|\hat\omega_j-\omega_j|
    \leq C\max_{1\leq j,k\leq p}|\hat r_{jk}-r_{jk}|
    =O_{\mathbb P}\bigg(\frac{\log^{3/2}p}{\sqrt n}\bigg).
\end{equation}
It follows that \(\min_j\hat\omega_j\) is bounded away from zero with probability tending to one. Since
\[
    \hat q_j
    =\frac{n^{-1}\sum_{i=1}^n(\hat g_{ij}-\bar{\hat g}_j)^2}{\hat\omega_j}
\]
on this event, the preceding bounds imply
\begin{equation}
\label{eq:delta_rate}
    \vartheta_{n,p}:=\max_{1\leq j\leq p}|\hat q_j-1|
    =O_{\mathbb P}\bigg(\frac{\log^{9/2}p}{\sqrt n}\bigg)
    +o(p^{-3/4}).
\end{equation}
Under \(\log p=o(n^{1/14})\),
\[
    \vartheta_{n,p}\log^2p
    =O_{\mathbb P}\bigg(\frac{\log^{13/2}p}{\sqrt n}\bigg)
    +o\big(p^{-3/4}\log^2p\big)
    =o_{\mathbb P}(1).
\]
On the event \(\{\vartheta_{n,p}<1/2\}\), both \(\max_j\hat q_j\) and \(\operatorname{med}_j\hat q_j\) belong to \([1-\vartheta_{n,p},1+\vartheta_{n,p}]\). Hence,
\[
    \bigg|\frac{\max_j\hat q_j}{\operatorname{med}_j\hat q_j}-1\bigg|
    \leq C\vartheta_{n,p}.
\]
For \(n\geq2\), the exponent used in the proposed procedure,
\[
    \gamma_{n,p}
    =\max\bigg\{\frac1{12},
    \min\bigg\{\frac18,\frac1{10}
    +\frac{\log(p/n)}{40\log n}\bigg\}\bigg\},
\]
belongs to \([1/12,1/8]\). For every fixed \(G<\infty\), the derivatives of \(x\mapsto x^{-\gamma}\) are uniformly bounded for \(0\leq\gamma\leq G\) and \(x\) in a neighborhood of one. Consequently, any possibly exponent \(\gamma_{n,p}\in[0,G]\), including the proposed choice, satisfies
\[
    |\hat a_{n,p}-1|\leq C\vartheta_{n,p},
    \qquad
    |\hat a_{n,p}-1|\log^2p=o_{\mathbb P}(1).
\]
Therefore \(0<\hat a_{n,p}\leq1\) for every sample. This completes the proof. \(\hfill\square\)

\bibliography{ref}

\end{document}